\documentclass[letterpaper, 10 pt, conference]{ieeeconf}
\IEEEoverridecommandlockouts
\usepackage{amsmath,amssymb,amsfonts}

\usepackage{amsthm}
\usepackage{graphicx}
\usepackage{textcomp}
\usepackage{xcolor}
\usepackage{hyperref}

\usepackage[noend]{algorithmic}
\usepackage{algorithm}

\usepackage{graphicx}
\usepackage{subfigure}

\usepackage{cite}

\usepackage[utf8]{inputenc}
\usepackage[english]{babel}
\newtheorem{problem}{Problem}

\newtheorem{proposition}{Proposition}
\newtheorem{definition}{Definition}

\allowdisplaybreaks[1]

\def\BibTeX{{\rm B\kern-.05em{\sc i\kern-.025em b}\kern-.08em
    T\kern-.1667em\lower.7ex\hbox{E}\kern-.125emX}}
\begin{document}

\title{Formal Verification of Stochastic Systems with ReLU \\ Neural Network Controllers
}
\author{
  Shiqi Sun,
  Yan Zhang,
  Xusheng Luo,
  Panagiotis Vlantis,
  Miroslav Pajic
  and
  Michael M. Zavlanos
  \thanks{
    Shiqi Sun,
    Yan Zhang,
    Xusheng Luo,
    Panagiotis Vlantis
    and
    Michael
    M. Zavlanos
    are with the
    Department
    of
    Mechanical
    Engineering
    and
    Materials
    Science,
    Duke
    University,
    Durham,
    NC
    27708,
    USA.
    \texttt{
      \small
      \{shiqi.sun,yan.zhang2,xusheng.luo, panagiotis.vlantis,michael.zavlanos\}@duke.edu
    }
    Miroslav Pajic
    is with the
    Department of Electrical and Computer Engineering
    Duke
    University,
    Durham,
    NC
    27708,
    USA.
    \texttt{
      \small
      {miroslav.pajic}@duke.edu
    }.
    This work is supported in part by AFOSR under award \#FA9550-19-1-0169
    and by NSF under award CNS-1932011.
  }
}
\maketitle

\begin{abstract}
  In this work,
  we address the problem of formal safety verification for
  stochastic cyber-physical systems (CPS) equipped with
  ReLU neural network (NN) controllers.
  Our goal is to find the set of initial states from where,
  with a predetermined confidence,
  the system will not reach an unsafe configuration
  within a specified time horizon.
  Specifically,
  we consider discrete-time LTI systems with Gaussian noise,
  which we abstract by a suitable graph.
  Then,
  we formulate  
  a Satisfiability Modulo Convex (SMC) problem
  to estimate upper bounds on
  the transition probabilities between nodes in the graph.
  %
  %
  Using this abstraction,
  we propose a method to compute tight bounds on
  the safety probabilities of nodes in this graph,
  despite
  possible over-approximations of the transition probabilities
  between these nodes.
  Additionally,
  using the proposed SMC formula,
  we devise a heuristic method to refine the abstraction of the system
  in order to further improve the estimated safety bounds.
  Finally,
  we corroborate the efficacy of the proposed method with
  simulation results
  considering a robot navigation example
  and
  comparison against a state-of-the-art verification scheme.
\end{abstract}


\section{Introduction}%
\label{sec/intro}

In recent years,
advances in the field of deep learning have furnished
a new class of adept and adaptable control schemes for cyber-physical systems
which considerably simplify the overall design process.
Robot navigation is one such exemplar application
where neural network controllers have been successfully employed for
steering a variety of robotic platforms in a diversity of situations~%
\cite{b3,b4,b5,b6,b7}.
As these AI-enabled systems open up new possibilities for control,
which are still considered understudied in the literature,
the issues of safety and reliability of neural network controllers,
become more pressing.
In addition,
as such control schemes are employed to address
safety-critical real-world problems,
the ability to formally verify the security of the neural networks becomes
imperative~%
\cite{b8}.

To address these challenges,
a significant effort has been recently directed to
the robustification and verification of deep neural networks.
Considering the former direction,
Generative Adversarial Network (GAN) methodologies
have been successfully employed to train networks with
improved tolerance to disturbances
\cite{b9,b10}.
Although these methods may yield noticeably more robust networks,
they provide no means of estimating the reliability of the system.
On the other hand,
verification schemes provide ways to estimate bounds on the output of
already trained networks and answer reachability queries related to
the corresponding closed-loop dynamics.
%
In \cite{b11,b12},
a reachability analysis method for neural networks
is proposed that relies on
semi-definite programming whereas,
in \cite{b13},
satisfiability modulo theory is employed to provide
formal verification guarantees.
%
In \cite{b14,b15},
a hybrid system verification scheme is proposed to
answer reachability queries concerning dynamical systems equipped with
neural network controllers.
Likewise,
the Satisfiability Modulo Convex optimization (SMC) approach~\cite{b17}
is adopted in~\cite{b1} in order to
verify the safety of neural networks with ReLU activation functions
for robotic platforms equipped with proximity sensors.
Stochastic problems are also considered in \cite{b20,b21,b22,b23,b24}.
%
%
Particularly,
in~\cite{b23,b24}
a sampling-based method is proposed to ensure
safety of a closed-loop system
subject to randomness only in the initial conditions,
whereas
the methodologies in~\cite{b21,b22}
are limited to mixed monotone stochastic systems.
Finally,
guarantees on the safety of stochastic switched systems
equipped with general nonlinear controllers are also derived in~\cite{b20}
using
Internal Markov Chain (IMC)
and
Bound-Parameter Markov Decision Process (BMDP) methods
which, however, require very fine partitions of the domain in order to
furnish accurate safety probability bounds.

In this paper,
we propose a new verification scheme for
stochastic dynamical systems equipped with ReLU neural network controllers.
Particularly,
we consider discrete-time LTI systems with Gaussian noise and
partition of the continuous state space into convex sets (e.g.,
constructed as in \cite{b1}).
Then,
we abstract the system by a graph
and
formulate a Satisfiability Modulo Convex problem
which we solve using existing tools
in order to
estimate valid upper bounds on
the transition probabilities between pairs of nodes in the graph.
Using this transition graph,
we also develop an algorithm
to estimate tight upper bounds on
the probability the system reaches the set of unwanted states
after a specified amount of steps,
even when the underlying transition probability bounds
have been over-estimated.
%
%
Unlike methodologies such as~\cite{b23,b24} that yield
probably correct estimations of the safety probability bounds,
here we provide bounds that are correct by design.
Additionally,
we use the proposed SMC formula to devise a heuristic method to
subdivide the cells in a given abstraction in order to
further improve our safety probability estimations.
Finally,
we provide numerical simulations on a robot navigation problem
corroborating the efficacy of
our proposed verification method compared to~\cite{b20},
which provides looser bounds on the safety probability for
the coarse partitions considered here.

We organize the paper as follows.
%
In \autoref{sec/prob},
we formulate the problem under consideration
while
in \autoref{sec/graph}
we elaborate on
the construction of the graph
and
the methodology used for
computing upper bounds on the transition probabilities.
Then,
in \autoref{sec/safety},
we present the proposed verification scheme
and
%
in \autoref{sec/refine}
we develop the proposed heuristic method to refine
the selected state abstraction.
Finally,
in \autoref{sec/sim}
we conclude this work by presenting comparative results corroborating the
efficacy of our scheme.

\section{Problem Formulation}%
\label{sec/prob}


We consider an autonomous robot moving in
a compact, polytopic workspace $ \mathcal{W}\subset\mathbb{R}^p $
occupied by
a set of zero or more inner obstacles $ \{\mathcal{O}_i\}_{i=1}^{p_o}$.
Let \(\mathcal{W}_{s} = \mathcal{W} \setminus \cup_{i=1}^{p_o} \mathcal{O}_{i} \)
be the set of safe robot positions and let \( \mathcal{W}_{o} \) denote
its complement.
The robot's dynamics are described by the following
linear discrete-time stochastic model
\begin{equation}\label{eq/sys}
    x^{t+1}=Ax^{t}+Bu^{t}+w^t,
\end{equation}
where
\( x^t \in \mathcal{X} \subseteq \mathbb{R}^{n} \)
and
\( u^t \in \mathcal{U} \subset \mathbb{R}^{m} \)
denote the robot's state and control input at time \( t \),
respectively,
and
$w^t\sim\mathcal{N}(0, \sigma_t)$, $\sigma_t \in \mathbb{R}^n$
is externally induced Gaussian noise, applied at time \( t \).
We assume that the robot is equipped with one or more sensors that
allow it to perceive the unoccluded region of the workspace around it.
We shall use \( d(x^t) \) to denote
the sensor measurements obtained at configuration \( x^t \),
with
\( d: \mathbb{R}^n \mapsto \mathbb{R}^{q} \) being the measurement function.
Also,
we assume that a pre-trained, feed-forward neural network controller
\( f_{\mathrm{NN}}: \mathbb{R}^{q} \mapsto \mathbb{R}^{m} \)
is provided for steering the robot to a desired configuration, i.e.,
\( u^t = f_{\mathrm{NN}}(d(x^t)) \).
Particularly,
we assume that the controller consists of \( L \) fully connected layers, i.e.,
\begin{equation}\label{eq/ctrl}
  \begin{split}
    h^{1}&=\max\big(0,W_{\phi}^0d(x^{t})+w_{\phi}^0\big),\\
    h^{2}&=\max\big(0,W_{\phi}^1h^{1}+w_{\phi}^1\big),\\
    &\hspace{3em} \vdots\\
    h^{L}&=\max\big(0,W_{\phi}^{L-1}h^{L-1}+w_{\phi}^{L-1}\big),\\
    u^t&=W_{\phi}^Lh^{L}+w_{\phi}^L
  \end{split},
\end{equation}
where
$W_{\phi}^l\in\mathbb{R}^{M_{i} \times M_{l-1}}$,
$w_{\phi}^l\in\mathbb{R}^{M_l}$
are pre-trained weight bias matrices
and
$h^{i}$ denotes the output of the \(i\)-th layer.


Given the stochastic system \eqref{eq/sys}
and
associated control law \eqref{eq/ctrl},
let
%
\( P_{k}: \mathbb{R}^n \mapsto [0,1] \)
denote the probability that the robot will reach
an unsafe state after \(k\) time steps
starting from \( x^t \), i.e.,
\begin{equation}
  P_{k}(x^t) = P( \mathcal{P}_{W}(x^{t+k}) \in \mathcal{W}_{o} ~|~ x^t ),
\end{equation}
where
\( \mathcal{P}_{W}: \mathbb{R}^{n} \mapsto \mathbb{R}^{p} \)
is a projection operator that returns the robot's current position.
In the remainder,
we shall say that a state \( x^{t} \) is
\((p,k)\)-safe if \( P_{k}(x) \leq p \), given \( p \in [0, 1] \).
%
Note that, in practice,
computing a precise approximation of \( P_{k} \)
may generally be intractable.
Therefore,
in this work,
we address the problem of computing
a correct-by-design tight upper bound on \( P_{k} \),
which allows us to answer safety queries,
albeit more conservatively.
%
\begin{problem}
  \label{prob/main}
  Given a robotic system obeying the closed-loop dynamics
  \eqref{eq/sys} and \eqref{eq/ctrl},
  compute a tight upper bound of the probability function \( P_{k}(x) \),
  for given horizon \( k \) and \( \forall x \in \mathcal{X} \).
\end{problem}

\section{Transition Graph}%
\label{sec/graph}

%
%
%
%
%
%

In order to address \autoref{prob/main},
in this section
we develop a methodology
to construct a discrete abstraction of the system's dynamics
and
to compute upper bounds on the transition probabilities between
different pairs of cells in this abstraction.
Then, in \autoref{sec/safety},
we utilize this transition graph to
estimate upper bounds on the safety probability \( P_{k} \).


%
%

We begin by partitioning\footnote{
  We assume that
  the state space \( \mathcal{X} \) consisting of the viable configurations
  (i.e., states where the robot does not overlap with the obstacles) is
  either given as or sufficiently approximated by a polytope.
} %
the state space \( \mathcal{X} \)
into a set \( \mathcal{S} = \{\mathcal{S}_{i}\}_{i=1}^{p_s} \) of
\( p_s \) non-overlapping convex polytopes
such that
$\mathcal{X} = \cup_{i=1}^{p_s} \mathcal{S}_{i}$
and
$\mathcal{S}_{i} \cap \mathcal{S}_{j} = \emptyset$
for all \(i \neq j \).
Let \( \mathcal{E} \subseteq \mathcal{S} \times \mathcal{S} \)
consist of all the pairs \( \mathcal{S}_i, \mathcal{S}_j \) such that
there exists at least one \( x^{t} \in \mathcal{S}_i \) for which
\( P(x^{t+1} \in \mathcal{S}_j) > 0 \).
%
Using the partition \( \mathcal{S} \),
we can model the dynamics of the stochastic system in~\eqref{eq/sys} by
a graph with edge weights that correspond to upper bounds on
the transition probabilities between all pairs of
abstract states \(\mathcal{S}_i\) and \(\mathcal{S}_j\).
Specifically, we state the following definition.
\begin{definition}
  \label{def/graph}
  Given
  the system \eqref{eq/sys}, \eqref{eq/ctrl}
  and
  the partition of the space \( \mathcal{S} \),
  an Upper Bound Probabilistic Transition Graph is
  a tuple $\mathcal{D} = (\mathcal{S}, \mathcal{E}, \hat{P})$ such that
  \(
  P(x^{t+1} \in \mathcal{S}_j | x^{t} \in \mathcal{S}_i ) \leq
  \hat{P}( \mathcal{S}_i, \mathcal{S}_j )
  \)
  for all
  \( (\mathcal{S}_i, \mathcal{S}_j) \in \mathcal{E} \).
\end{definition}
\noindent%
In order to construct the transition graph in~\autoref{def/graph},
we require a function
$\hat{P}:\mathcal{S}\times\mathcal{S}\rightarrow[0, 1]$
that upper bounds the transition probability from
state \( \mathcal{S}_i \) to \( \mathcal{S}_j \) from above.
To accomplish this,
we extend the SMC encoding presented in~\cite{b1}
to the case of stochastic dynamical systems considered.
%
%

Specifically,
we consider the evolution \( X^{t+1} = A X^{t} + B u^{t} \),
where \( X^{t} \) is the expectation of \( x^{t} \).
Notice that \( x^{t+1} \) is normally distributed, i.e.,
\( x^{t+1} \sim \mathcal{N}(X^{t+1}, \delta^{t+1}) \),
for all \( X^{t} \in \mathcal{X} \).
Notice also that
any given convex polytope \( \mathcal{S}_i \)
can be defined as the intersection of a finite number of hyperplanes
\( a_{i, \ell}^T x \leq b_{i, \ell},~ \ell \in \mathcal{G}(\mathcal{S}_i) \),
with \( \mathcal{G}(\mathcal{S}_i) \) some arbitrary indexing.
%
%
%
Following a procedure similar to the one in~\cite{b18},
we can compute the augmented set
\begin{equation}%
  \label{eq/stateset/augmented}
  \Bar{\mathcal{S}}_i(q) = \{
  X |
  \min_{\ell \in \mathcal{G}(\mathcal{S}_i)}(P( V_{i,\ell}(x) \leq 0)) \geq q,
  ~
  x \sim (X, \delta)
  \},
\end{equation}
with
\(
V_{i,\ell}(x) = a_{i,\ell} x - b_{i,\ell},
\ell \in \mathcal{G}(\mathcal{S}_i)
\),
which is convex
and
also
its complement consists only of states \( X^{t+1} \) such that
\(
P( x^{t+1} \in \mathcal{S}_i | x^{t+1} \sim \mathcal{N}(X^{t+1}, \delta^{t+1}) )
< q
\).
%
This last fact can be derived by noticing that
\(
P(x^{t+1} \in \mathcal{S}_i)
<
\min_{\ell \in \mathcal{G}(\mathcal{S}_i)}(P(a_{i,\ell}^T x^{t+1} \leq b_{i,\ell}))
\).
%
Next,
let $b^l_i$ indicate the activation status of
the $i$-th node in the $l$-th layer of
the neural network controller \( f_{\mathrm{NN}} \), i.e.,
$b^l_i$ is false when $h^l_i = 0$.
Then,
given a probability threshold \( q \in [0, 1] \)
and
a pair of \( \mathcal{S}_i, \mathcal{S}_j \in \mathcal{S} \),
we can define the following SMC problem,
which we refer to as Stochastic Neural Network SMC (SNN-SMC):
\begin{align}
  & \exists~
    X^t, X^{t+1} \in \mathbb{R}^n,
    u^t \in \mathbb{R}^m,
    d\in \mathbb{R}^{2N}
  \\ & (b^l,h^l,t^l)
       \in
       \mathbb{B}^{M_l}\times\mathbb{R}^{M_l}\times\mathbb{R}^{M_l}
       \nonumber
\end{align}
subject to:
\begin{align}
  & X^t\in\mathcal{S}_i
    \label{eq/smc/state/prev}
  \\ & \land
       X^{t+1}\in\Bar{\mathcal{S}_j}(q)
       \label{eq/smc/state/next}
  \\ & \land
       {X^{t+1}}=A X^t+ B u
       \label{eq/smc/dynamics}
  \\ & \land
       \big(t^1=W_{\phi}^0d(X^t)+w_{\phi}^0\big)
       \cap
       \big(\overset{L}{\underset{l=2}{\bigwedge}}
       t^l=W_{\phi}^{l-1}h^{l-1}+w_{\phi}^l\big)
       \label{eq/smc/nn/hidden}
  \\ & \land
       \big(u^t=W_{\phi}^Lh^L+w_{\phi}^L\big)
       \label{eq/smc/nn/out}
  \\ & \land
       \overset{L}{\underset{l=1}{\bigwedge}}
       \overset{M_i}{\underset{i=1}{\bigwedge}}
       b^l_i \to \big[\big(h^l_i=t^l_i\big)\land\big(t^l_i\geq 0\big)\big]
       \label{eq/smc/bool1}
  \\ & \land
       \overset{L}{\underset{l=1}{\bigwedge}}
       \overset{M_i}{\underset{i=1}{\bigwedge}}
       \neg b^l_i\to\big[\big(h^l_i=0\big)\land\big(t^l_i< 0\big)\big].
       \label{eq/smc/bool2}
\end{align}
In the above definition,
\eqref{eq/smc/state/prev},\eqref{eq/smc/state/next}
encode the transition from state $\mathcal{S}_i$ to $\mathcal{S}_j$,
\eqref{eq/smc/dynamics}
encodes the system dynamics \eqref{eq/sys},
and
\eqref{eq/smc/nn/hidden}-\eqref{eq/smc/bool2}
encode the behavior imposed by the neural network controller.

Using this SNN-SMC encoding,
we now present our proposed algorithm to compute
the upper bounds \( \hat{P} \) in \autoref{def/graph}
on the underlying transition probabilities.
Specifically,
observe that for
a given pair of cells \( \mathcal{S}_i, \mathcal{S}_j \)
if the selected threshold \( q \) is large,
then there may not be \( x^{t} \) such that
\( P( x^{t+1} \in \mathcal{S}_j | x^t  \in \mathcal{S}_i) \geq q \),
which would render the SNN-SMC problem unsatisfiable.
Therefore,
any such threshold \( q \)
is a valid upper bound on the transition probability
from \( \mathcal{S}_i \) to \( \mathcal{S}_j \), i.e.,
\( P( \mathcal{S}_i, \mathcal{S}_j ) \leq q \).
%
%
Based on this fact,
we propose an iterative algorithm
outlined in \autoref{alg/prob/transition},
which, given a user-specified precision \( dq \in (0, 1) \),
employs binary search for finding \( q \) such that
the SNN-SMC problem for
\( \mathcal{S}_i, \mathcal{S}_j \) can no longer be satisfied.
%
%
By executing \autoref{alg/prob/transition} for
every pair of cells \( \mathcal{S}_i \) and \( \mathcal{S}_j \)
in the partition \(  \mathcal{S} \),
we can obtain the desired function \( \hat{P} \).

\begin{algorithm}[t]
  \hspace*{0.02in}
  {\bf Input:} $\mathcal{S}_i$, $\mathcal{S}_j$, $dq$ \\
  \hspace*{0.02in}
  {\bf Output:} $\hat{P}(\mathcal{S}_i,\mathcal{S}_j)$
  \begin{algorithmic}[1]
    \STATE $q_{l} \gets 0$
    \STATE $q_{r} \gets 1$
    \WHILE{$q_{r} - q_{l} > dq$}
    \STATE $q \gets 0.5 (q_{l} + q_{l})$
    \IF{ $\text{SNN-SMC}(\mathcal{S}_i,\mathcal{S}_j,q)$ not satisfiable }
    \STATE $q_{r} \gets q$
    \ELSE
    \STATE $q_{l} \gets q$
    \ENDIF
    \ENDWHILE
    \RETURN $q_{r}$
  \end{algorithmic}
  \caption{Estimation of $\hat{P}(\mathcal{S}_i, \mathcal{S}_j)$}
  \label{alg/prob/transition}
\end{algorithm}

\section{Safety Probability Bounds}%
\label{sec/safety}

In this section,
we elaborate on how to compute upper bounds on
the safety probability \( P_{k} \)
given a valid transition graph \( \mathcal{D} \).
%
First,
we extend the definition of a $(p,k)$-safe state $x$ to
a $(p,k)$-safe cell.
Specifically,
we say that a cell $\mathcal{S}_i$ is $(p,k)$-safe if
all the states in $\mathcal{S}_i$ are $(p,k)$-safe.
Let
\( \hat{P}_{k}: \mathcal{S} \mapsto [0,1] \)
denote any function that
bounds the safety probability \( P_{k} \) from above, i.e.,
\( \hat{P}_k(\mathcal{S}_i) \geq \max_{x \in \mathcal{S}_i}( P_{k}(x) ) \)
for all \( \mathcal{S}_i \in \mathcal{S} \).
Then,
it is simple to see that any cell $\mathcal{S}_i \in \mathcal{S}$
is $(p,k)$-safe if
$\hat{P}_k(\mathcal{S}_i) \leq p$.
Let
\(
\mathcal{N}_{\mathcal{S}_i} =
\{ \mathcal{S}_j | (\mathcal{S}_i, \mathcal{S}_j) \in \mathcal{E} \}
\)
denote the set of nodes \( \mathcal{S}_j \) reachable from \( \mathcal{S}_i \).
Next,
we discuss how to compute
tight safety probability bounds $\hat{P}_k$ for all cells
in the partition \( \mathcal{S} \).
%
%
We begin by presenting the following proposition which
provides a straightforward method to compute \( \hat{P}_{k+1} \)
given \( \hat{P}_{k} \).




\begin{proposition}%
  \label{stmt/safety/obvious}
  Assume that the transition and safety probability bounds
  \( \hat{P} \) and \( \hat{P}_k \) are known.
  Then, for any \( \mathcal{S}_i \in \mathcal{S} \) the following holds:
  \begin{equation}
    P_{k+1}( x ) \leq
    \sum_{\mathcal{S}_j \in \mathcal{N}_{\mathcal{S}_i}} {
      \hat{P}_{k}( \mathcal{S}_i ) \hat{P} ( \mathcal{S}_i, \mathcal{S}_j )
    }
    ,~~~ \forall x \in \mathcal{S}_i
  \end{equation}
\end{proposition}
\begin{proof}
  Let \( x^t \in \mathcal{S}_i \). Then, by definition, it holds that
  \begin{equation*}
    \begin{aligned}
      &
      P_{k+1}(x^{t}) =
      P(x^{t+k+1} \in \mathcal{W}_{o} | x^{t} \in \mathcal{S}_{i})
      \\ & ~ =
      \sum_{\mathcal{S}_j \in \mathcal{N}_{\mathcal{S}_i}} {
        P(x^{t+k+1} \in \mathcal{W}_{o} | x^{t+1} \in \mathcal{S}_{i})
        P(x^{t+1} \in \mathcal{S}_{i} | x^{t} \in \mathcal{S}_{j})
      }
      \\ & ~ \leq
      \sum_{\mathcal{S}_j \in \mathcal{N}_{\mathcal{S}_i}} {
        \hat{P}_{k}(\mathcal{S}_{j})
        \hat{P}(\mathcal{S}_{j}, \mathcal{S}_{i})
      }
    \end{aligned}
  \end{equation*}
\end{proof}
\noindent
We remark that
a valid choice of \( \hat{P}_{0} \) is to let
\( \hat{P}_{0}(\mathcal{S}_i) = 1 \) if
there exists \( x \in \mathcal{S}_i \) such that \( P_W(x) \in \mathcal{W}_o \)
and
\( \hat{P}_{0}(\mathcal{S}_i) = 0 \) otherwise.
Using this choice for \(\hat{P}_{0}\),
we can recursively compute the safety probability
\( \hat{P}_{k}(\mathcal{S}_i) \) for
every cell \(\mathcal{S}_i \in \mathcal{S}\) as
\begin{equation}%
  \label{eq/spb/naive}
  \hat{P}_{k+1}(\mathcal{S}_i) =
  \hat{P}'_{k+1}(\mathcal{S}_i) \triangleq
  \sum_{\mathcal{S}_j \in \mathcal{N}_{\mathcal{S}_i}} {
    \hat{P}_{k}( \mathcal{S}_i ) \hat{P} ( \mathcal{S}_i, \mathcal{S}_j )
  }
\end{equation}
until the desired horizon \( k = T \) is reached.
However,
note that~\eqref{eq/spb/naive}
is expected to furnish loose bounds on the safety probabilities when
the estimations of the underlying transition probabilities are not tight
or
the partition is coarse.
The reason for the latter case is that
\autoref{alg/prob/transition} computes
a worst case upper bound on the transition probability from
cell \( \mathcal{S}_i \) to cell \( \mathcal{S}_j \)
that is close to the transition probability from the worst case state
\(
X^{\star} = \mathrm{argmax}_{X^t \in \mathcal{S}_i}{
  P(x^{t+1} \in \mathcal{S}_j | X^{t} ))
}
\)
with
\( x^{t+1} \sim \mathcal{N}(X^{t+1}, \delta^{t+1}) \)
and
\( X^{t+1} = A X^{t} + B f_{\mathrm{NN}}(d^{t}(X^t)) \).
%
%
%
%
Note that if the partition is coarse,
it is likely that there are many other states in \(\mathcal{S}_i\) that have
much lower transition probabilities to \(\mathcal{S}_j\)
but are effectively treated the same as \( X^{\star} \).
If the partition is finer,
many of these states can be grouped in a different cell with
lower transition probability to the cell \(\mathcal{S}_j\).
We discuss a way to refine the partition \(\mathcal{S}\)
in \autoref{sec/refine}.
But first we describe how to improve
the safety probability bounds \(\hat{P}_k\) in~\eqref{eq/spb/naive}
for a given partition \(\mathcal{S}\).
%

%

\subsection{Partition Merging}%
\label{sec/safe/merge}

In this section,
we present a method to merge nodes in \( \mathcal{S} \)
in order to improve the transition probability bounds
used in~\eqref{eq/spb/naive}.
Specifically, we provide the following result.
\begin{proposition}%
  \label{stmt/safety/merging}
  Let \( \mathcal{S}_o \in \mathcal{S} \),
  \(\mathcal{S}_i, \mathcal{S}_j \in \mathcal{N}_{\mathcal{S}_o}\)
  and
  \( p \in (0, 1) \).
  %
  %
  Consider the set \( \mathcal{S}_{ij}' = \mathcal{S}_i \cup \mathcal{S}_j \).
  If
  \(
  \Bar{\mathcal{S}}_i(p)
  \cap
  \Bar{\mathcal{S}}_j(p) =
  \emptyset
  \),
  then
  \( \mathcal{D}' = (\mathcal{S}', \mathcal{E}', \hat{P}') \)
  with
  \( \mathcal{S}' = \mathcal{S} \cup \{ \mathcal{S}_{ij}' \} \),
  \(
  \mathcal{E}'
  =
  \mathcal{E}
  \bigcup
  \{ (\mathcal{S}_o, \mathcal{S}_{ij}') \}
  \setminus
  \{
  (\mathcal{S}_o, \mathcal{S}_{i}),
  (\mathcal{S}_o, \mathcal{S}_{j})
  \}
  \),
  \(
  \hat{P}'(\mathcal{S}_o, \mathcal{S}_{ij}' ) =
  \max(
  \max(
  \hat{P}(\mathcal{S}_o, \mathcal{S}_{i}),
  \hat{P}(\mathcal{S}_o, \mathcal{S}_{j})
  ) + p
  , 2 p
  )
  \)
  and
  \(
  \hat{P}'(\mathcal{S}_o, \mathcal{S}_{i}) = 
  \hat{P}'(\mathcal{S}_o, \mathcal{S}_{j}) = 0
  \)
  is a valid transition graph.
\end{proposition}
\begin{proof}
  Let \( q_i, q_j \in (0, 1) \).
  Recall that the complement of
  \( \Bar{\mathcal{S}}_i(q_i) \)
  (resp. \( \Bar{\mathcal{S}}_j(q_j) \))
  consists only of states \( X \in \mathcal{X} \) such that
  \( P(x \in \mathcal{S}_i ) < q_i \)
  (resp. \( P(x \in \mathcal{S}_j ) < q_j \))
  for \( x \sim \mathcal{N}(X, \delta) \).
  %
  Let \( X^{t} \) be a state in \( \mathcal{S}_o \).
  Assuming 
  \(
  \Bar{\mathcal{S}}_i(q_i)
  \cap
  \Bar{\mathcal{S}}_j(q_j) =
  \emptyset
  \),
  then
  \( X^{t+1} \) must lie either in
  \( \Bar{\mathcal{S}}_i(q_i) \),
  \( \Bar{\mathcal{S}}_j(q_j) \),
  or
  the complement of their union.
  As such, the following inequalities hold:
  \begin{equation}
    P(x^{t+1} \in \mathcal{S}_{ij}' ) <
    \begin{cases}
      \hat{P}( \mathcal{S}_o, \mathcal{S}_i ) + q_j,
      & \hspace{-0.5em}
      \text{ if } X^{t+1} \in \Bar{\mathcal{S}}_i(q_i)
      \\
      \hat{P}( \mathcal{S}_o, \mathcal{S}_j ) + q_i,
      & \hspace{-0.5em}
      \text{ if } X^{t+1} \in \Bar{\mathcal{S}}_j(q_j)
      \\
      q_i + q_j,
      & \hspace{-0.5em}
      \text{ if } X^{t+1} \in
      \mathcal{X} \setminus \{\Bar{\mathcal{S}}_j, \Bar{\mathcal{S}}_j \}.
    \end{cases}
  \end{equation}
  Setting \( q_i = q_j = p \) and choosing the worst of these cases
  concludes the proof.
\end{proof}
%
\autoref{stmt/safety/merging}
introduces a new node \( \mathcal{S}'_{ij} \) to
the transition graph \( \mathcal{D}\)
that is the result of ``merging''
cells \( \mathcal{S}_i, \mathcal{S}_j \in \mathcal{N}_{\mathcal{S}_o} \)
which are far enough from each other so that
their augmented sets
\( \Bar{\mathcal{S}}_i(p) \) and \( \Bar{\mathcal{S}}_j(p) \)
do not overlap,
where \( p \in (0, 1) \) is a user specified probability threshold.
%
The algorithm is illustrated in \autoref{alg/safety/merging}.
Notice that
neither \( \mathcal{S}_i \) nor \( \mathcal{S}_j \)
are removed from \( \mathcal{D} \);
what is removed is their edges with \( \mathcal{S}_o \).
Consequently,
by repeatedly merging cells in the graph \( \mathcal{D} \)
until there are no more cells that can be merged,
one can obtain a new graph \( \mathcal{D}' \)
which has more nodes than \( \mathcal{D} \)
but is not fully connected, i.e.,
\( | \mathcal{N}'_{\mathcal{S}_o}| \leq | \mathcal{N}_{\mathcal{S}_o}| \)
for all \( \mathcal{S}_o \in \mathcal{S} \subseteq \mathcal{S}' \).
Note also that
\(
\max_{X^t \in \mathcal{S}_o}{
  P(x^{t+1} \in \mathcal{S}_i \cup \mathcal{S}_j | X^{t})
}
\)
is bounded from above by
\(
\max_{X_i^t \in \mathcal{S}_o}{ P(x_{i}^{t+1} \in \mathcal{S}_i | X_{i}^{t} ) }
+
\max_{x_j^t \in \mathcal{S}_o}{ P(x_{j}^{t+1} \in \mathcal{S}_j | X_{j}^{t} ) }
\)
for all cells \( \mathcal{S}_i, \mathcal{S}_j \in \mathcal{N}_{\mathcal{S}_o} \)%
\footnote{
  We recall that
  \( x^{t+1} \) is a random variable sampled from
  the normal distribution \( \mathcal{N}(X^{t+1}, \delta^{t+1}) \)
  and
  the probability of \( x^{t+1} \) lying in a set \( \mathcal{S}_o \)
  is obtained by
  integrating the corresponding density function over \( \mathcal{S}_o \).
}.
%
%
In addition,
\autoref{stmt/safety/merging} informs us that
\(
\hat{P}(\mathcal{S}_o, \mathcal{S}'_{ij})
\ll
\hat{P}(\mathcal{S}_o, \mathcal{S}_i)
+
\hat{P}(\mathcal{S}_o, \mathcal{S}_j)
\)
for those cells \( \mathcal{S}_i, \mathcal{S}_j \) which
are reachable from \( \mathcal{S}_o \) with high transition probabilities
but
are far away from each other,
given small probability threshold \( p \),
because
there exists no single state \( X^{t} \) in \( \mathcal{S}_o \)
which is as likely to transition to \( \mathcal{S}_i \)
as to \( \mathcal{S}_j \) (see \autoref{fig/merging}).
%
%
Finally,
in order to guarantee that the new bounds we obtain on
the safety probabilities \(P_k\) are tighter than the original ones,
two cells \( \mathcal{S}_i \) and \( \mathcal{S}_j \)
are merged only if the following condition holds:
\begin{equation}
  \label{eq/merge/condition}
  \begin{aligned}
    &P(\mathcal{S}_o, \mathcal{S}'_{ij})
    \max( \hat{P}_{k}(\mathcal{S}_i), \hat{P}_{k}(\mathcal{S}_j) )
    <
    \\ &\quad\quad\quad\quad
    P(\mathcal{S}_o, \mathcal{S}_i) \hat{P}_{k}(\mathcal{S}_i)
    +
    P(\mathcal{S}_o, \mathcal{S}_j) \hat{P}_{k}(\mathcal{S}_j).
  \end{aligned}
\end{equation}
Particularly,
this condition ensures that the terms of~\eqref{eq/spb/naive}
effectively removed by \autoref{alg/safety/merging} get replaced by
strictly lesser ones.

\begin{figure}[]
  \centering
  \begin{minipage}[t]{1\linewidth}
    \centering
    \includegraphics[height=4.0cm]{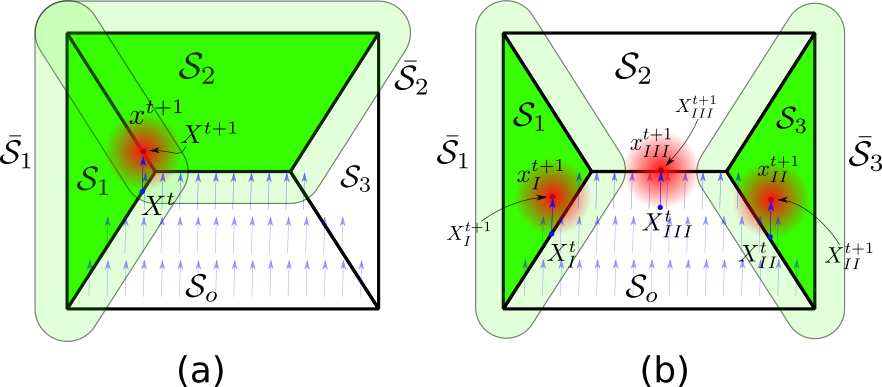}
  \end{minipage}%
  \caption{
    Example of mergeable and not mergeable pairs of cells
    according to \autoref{stmt/safety/merging}.
    The probability density function of each \( x^{t+1} \) is depicted using
    red-colored gradient.
    a)
    Cells \( \mathcal{S}_{1} \) and \( \mathcal{S}_{2} \) are not mergeable
    because the intersection of their augmented sets
    \( \Bar{\mathcal{S}_{1}} \) and \( \Bar{\mathcal{S}_{2}} \) is not empty
    and there exists \( X^{t} \in \mathcal{S}_{o} \)
    whose \( X^{t+1} \) lies in
    \( \Bar{\mathcal{S}}_{1} \cap \Bar{\mathcal{S}}_{2} \).
    As such,
    it is not possible to deduce a tight bound on the probability
    \( P(\mathcal{S}_{1} \cup \mathcal{S}_2 | \mathcal{S}_{o}) \)
    based on readily available bounds on
    \( P(\mathcal{S}_{1} | \mathcal{S}_{o}) \)
    and
    \( P(\mathcal{S}_{2} | \mathcal{S}_{o}) \);
    in this particular case,
    \( P(\mathcal{S}_{1} \cup \mathcal{S}_2 | \mathcal{S}_{o}) \approx 1 \).
    b)
    Cells \( \mathcal{S}_{1} \) and \( \mathcal{S}_{3} \) can be merged and
    a valid bound on the transition probability
    \( P(\mathcal{S}_1 \cup \mathcal{S}_3 | \mathcal{S}_o) \)
    is approximately equal to
    \( \max{( P(\mathcal{S}_1 | \mathcal{S}_o), P(\mathcal{S}_3 | \mathcal{S}_o) )} \),
    since there is no \( X^{t} \in \mathcal{S}_{o} \) with
    a good probability of landing in \( \mathcal{S}_1 \)
    and
    a good probability of landing in \( \mathcal{S}_2 \).
  }
  \label{fig/merging}
\end{figure}

\begin{algorithm}[t]
  \caption{Partition Merging}%
  \label{alg/safety/merging}
  \hspace*{0.02in} {\bf Input:} $ %
  \mathcal{D} = (\mathcal{S}, \mathcal{E}, \hat{P}), %
  \mathcal{S}_o, %
  p, %
  \hat{P}_k %
  $ \\
  \hspace*{0.02in} {\bf Output:} $ %
  \mathcal{D}' = (\mathcal{S}', \mathcal{E}', \hat{P}'), %
  \hat{P}'_{k} %
  $ \\
  \begin{algorithmic}[1]
    \STATE $ \mathcal{S}' \gets \mathcal{S} $
    \STATE $ \mathcal{E}' \gets \mathcal{E} $
    \STATE $ \hat{P}'_{k} \gets \hat{P}_{k} $
    \FOR{$ \mathcal{S}_{i} \text{ in } \mathcal{S}$}
    \FOR{$ \mathcal{S}_{j} \text{ in } \mathcal{S} \setminus \{ \mathcal{S}_i \}$}
    \IF {$ \Bar{\mathcal{S}}_i(p) \cap \Bar{\mathcal{S}}_j(p) $}
    \STATE %
    $ \mathcal{S}'_{ij} \gets \mathcal{S}_i \cup \mathcal{S}_j $
    \STATE %
    $ \mathcal{P}'(\mathcal{S}_o, \mathcal{S}'_{ij})
    \hspace{-0.2em}
    \gets
    \hspace{-0.2em}
    \max(
    \hat{P}(\mathcal{S}_o, \mathcal{S}_i),
    \hat{P}(\mathcal{S}_o, \mathcal{S}_j),
    p ) + p %
    $
    \IF{ \eqref{eq/merge/condition} holds }
    \STATE %
    $ \mathcal{S}' \gets \mathcal{S}' \cup \{ \mathcal{S}'_{ij} \} $
    \STATE %
    $ \mathcal{E}' \gets \left(
      \mathcal{E}' \setminus %
      \{ (\mathcal{S}_{o}, \mathcal{S}_{i}), (\mathcal{S}_{o}, \mathcal{S}_{j}) \}
    \right) \cup {(\mathcal{S}_o, \mathcal{S}'_{ij})} $
    \ENDIF
    \ENDIF
    \ENDFOR
    \ENDFOR
  \end{algorithmic}
\end{algorithm}

\subsection{Transition Probability Normalization}%
\label{sec/safe/tpn}

Given the transition graph \( \mathcal{D} \)
constructed in \autoref{sec/safe/merge},
next
we present an alternative way to recursively compute
tight bounds \(\hat{P}_k\) on the safety probabilities
when the underlying transition probability bounds are over-approximated.
%
%
Specifically,
we provide the following result,
which is similar to the one derived in~\cite{b20},
to truncate the sum in~\eqref{eq/spb/naive}
while ensuring that the new estimation
remains a valid upper bound of \( P_{k} \).
\begin{proposition}%
  \label{stmt/safety/tpn}
  Let \( \mathcal{S}_o \in \mathcal{S} \)
  and
  \( \hat{\kappa}_o: \mathbb{N} \mapsto \mathbb{N} \)
  such that
  \(
  \hat{P}_k(\mathcal{S}_{\hat{\kappa}_{o}}(i))
  \leq
  \hat{P}_k(\mathcal{S}_{\hat{\kappa}_{o}}(j))
  \),
  for all $i \leq j$
  with
  \(
  \mathcal{S}_{\hat{\kappa}_o(i)},
  \mathcal{S}_{\hat{\kappa}_o(j)}
  \in \mathcal{N}_{\mathcal{S}_o}
  \).
  Also,
  let
  \( n = |\mathcal{N}_{\mathcal{S}_o}| \)
  and
  \( \hat{m} \) such that
  \(
  \sum_{i=\hat{m}+1}^{n} \hat{P}( \mathcal{S}_o, \mathcal{S}_{\hat{\kappa}_o(i)}) \leq 1
  \).
  %
  Then,
  \( P_{k+1}(x) \leq \hat{P}''_{k+1}(x) \) for all \( x \in \mathcal{S}_o \),
  where
  \begin{equation}%
    \label{eq/safety/tpn/0}
    \begin{aligned}
      \hat{P}''_{k+1}(x) &=
      \sum_{i=\hat{m}+1}^{n} {
        \hat{P}(\mathcal{S}_o, \mathcal{S}_{\hat{\kappa}_o(i)})
        \hat{P}_k(\mathcal{S}_{\hat{\kappa}_o(i)})
      }
      \\ &\quad\quad +
      \left(
        1 - \sum_{i=\hat{m}+1}^n {
          \hat{P}(\mathcal{S}_o, \mathcal{S}_{\hat{\kappa}_o(i)})
        }
      \right)
      \hat{P}_k(\mathcal{S}_{\hat{\kappa}_o(\hat{m})})
      ,
    \end{aligned}
  \end{equation}
  Additionally,
  if \( \hat{m} \) such that
  \(
  \sum_{i=\hat{m}}^{n} \hat{P}( \mathcal{S}_o, \mathcal{S}_{\hat{\kappa}_o(i)}) > 1
  \),
  then
  \( \hat{P}''_{k+1}(x) < \hat{P}'_{k+1}(x) \)
  for all \( x \in \mathcal{S}_o \).
\end{proposition}
\begin{proof}
  We begin by noticing that
  \begin{equation}%
    \label{eq/safety/tpn/1}
    \begin{aligned}
      P_{k+1}(x^t) &\leq
      \sum_{\mathcal{S}_i \in \mathcal{N}_{\mathcal{S}_{o}}}{
        P( \mathcal{S}_o, \mathcal{S}_i )
        \cdot
        P_{k}(\mathcal{S}_i)
      }
      ,~~~ \forall x^{t} \in \mathcal{S}_o,
    \end{aligned}
  \end{equation}
  where
  \(
  P(\mathcal{S}_{o}, \mathcal{S}_{i})
  \triangleq
  P(x^{t+1} \in \mathcal{S}_i | x^t \in \mathcal{S}_o)
  \)
  and
  \( P_{k} \)
  denote
  the lowest upper bounds on the safety probabilities, i.e.,
  \( P_{k}(\cdot) \leq \hat{P}_{k}(\cdot) \).
  We recall that
  \(
  \sum_{i=1}^{\hat{m}}{
    P( \mathcal{S}_o, \mathcal{S}_{\hat{\kappa}_o(i)} )
  }
  =
  1 - \sum_{i=\hat{m}+1}^{n}{
    P( \mathcal{S}_o, \mathcal{S}_{\hat{\kappa}_o(i)} )
  }
  \).
  By separating the r.h.s. of \eqref{eq/safety/tpn/1} into two sums,
  one consisting of the \( n-\hat{m} \) terms corresponding to
  the cells with
  the highest estimated safety probability bounds \( \hat{P}_{k} \)
  and
  one consisting of the rest,
  we get
  \begin{equation}%
    \label{eq/safety/tpn/2}
    \begin{aligned}
      P_{k+1}(x^t) &\leq
      \sum_{i=\hat{m}+1}^{n}{
        P( \mathcal{S}_o, \mathcal{S}_{\hat{\kappa}_o(i)} )
        \cdot
        P_{k}(\mathcal{S}_{\hat{\kappa}_o(i)})
      }
      \\ & \quad\quad +
      P_{k}(\mathcal{S}_{\hat{\kappa}_o(m)})
      \cdot
      \left(
        1 - \sum_{i=\hat{m}+1}^{n}{
          P( \mathcal{S}_o, \mathcal{S}_{\hat{\kappa}_o(i)} )
        }
      \right),
    \end{aligned}
  \end{equation}
  where
  \(
  m
  =
  \mathrm{argmax}_{i \in \{1, 2, \ldots, m\}}( P_{k}(\mathcal{S}_{\hat{\kappa}_o(i)}) )
  \).
  In general,
  \( m \leq \hat{m} \)
  but
  one can readily see that
  \(
  P_{k}(\mathcal{S}_{\hat{\kappa}_{o}(m)})
  \leq
  \hat{P}_{k}(\mathcal{S}_{\hat{\kappa}_{o}(m)})
  \leq
  \hat{P}_{k}(\mathcal{S}_{\hat{\kappa}_{o}(\hat{m})})
  \).
  Therefore,
  one can verify that
  the r.h.s. of the above inequality is bounded from above by
  \begin{equation}%
    \label{eq/safety/tpn/3}
    \begin{aligned}
      &
      \sum_{i=\hat{m}+1}^{n}{
        P( \mathcal{S}_o, \mathcal{S}_{\hat{\kappa}_o(i)} )
        \cdot
        \hat{P}_{k}(\mathcal{S}_{\hat{\kappa}_o(i)})
      }
      \\ &\quad +
      \hat{P}_{k}(\mathcal{S}_{\hat{\kappa}_o(\hat{m})})
      \cdot
      \left(
        1 - \sum_{i=\hat{m}+1}^{n}{
          P( \mathcal{S}_o, \mathcal{S}_{\hat{\kappa}_o(i)} )
        }
      \right).
    \end{aligned}
  \end{equation}
  %
  By subtracting \eqref{eq/safety/tpn/3} from
  the r.h.s. of \eqref{eq/safety/tpn/0}
  and recalling that
  \( \hat{P}(\mathcal{S}_o, \cdot) > P(\mathcal{S}_o, \cdot) \)
  and
  \(
  \hat{P}_{k}(\mathcal{S}_{\hat{\kappa}_o(i)})
  \geq
  \hat{P}_{k}(\mathcal{S}_{\hat{\kappa}_o(\hat{m})})
  \)
  by construction of \( \hat\kappa_o \),
  we can verify that
  \( \hat{P}''_{k+1}(x) \) is a valid upper bound of \( P_{k+1} \).
  To show that \( \hat{P}''_{k+1}(x) \leq \hat{P}'_{k+1}(x) \)
  when
  \(
  \sum_{i=\hat{m}}^{n} \hat{P}( \mathcal{S}_o, \mathcal{S}_{\hat{\kappa}_o}(i)) > 1
  \),
  we notice that this implies that
  \begin{equation}
    \begin{aligned}
      &
      \hat{P}(\mathcal{S}_{o}, \mathcal{S}_{\hat{\kappa}_o(\hat{m})}) >
      1 - 
      \sum_{i=\hat{m}+1}^{n} \hat{P}( \mathcal{S}_o, \mathcal{S}_{\hat{\kappa}_o}(i))
      \implies \\ &
      \hat{P}(\mathcal{S}_{o}, \mathcal{S}_{\kappa_o(\hat{m})})
      \hat{P}_{k}(\mathcal{S}_{\hat{\kappa}_o(\hat{m})})
      > \\ &\quad\quad\quad
      \left(
        1 - 
        \sum_{i=\hat{m}+1}^{n} \hat{P}( \mathcal{S}_o, \mathcal{S}_{\hat{\kappa}_o}(i))
      \right)
      \hat{P}_{k}(\mathcal{S}_{\hat{\kappa}_o(i)})
      \implies \\ &
      \sum_{i=1}^{\hat{m}} {
        \hat{P}(\mathcal{S}_{o}, \mathcal{S}_{\kappa_o(i)})
        \hat{P}_{k}(\mathcal{S}_{\hat{\kappa}_o(i)}) }
      > \\ &\quad\quad\quad
      \left(
        1 - 
        \sum_{i=\hat{m}+1}^{n} \hat{P}( \mathcal{S}_o, \mathcal{S}_{\hat{\kappa}_o}(i))
      \right)
      \hat{P}_{k}(\mathcal{S}_{\hat{\kappa}_o(\hat{m})})
    \end{aligned}
  \end{equation}
\end{proof}
%
\autoref{stmt/safety/tpn}
provides an alternative formula to~\eqref{eq/spb/naive}
to recursively estimate the safety probability bounds \(\hat{P}_k\)
that uses the \( n - \hat{m} \) cells
with the largest estimated safety probabilities \( \hat{P}_k \)
in order to obtain a tighter bound on the safety probability
by mitigating the over-approximation of
the transition probability bounds.

\subsection{Verification Framework}%
\label{sec/safety/framework}

Given a graph \( \mathcal{D} \),
upper bounds on
the transition and safety probabilities
and a user specified horizon T,
we now present an algorithm to compute \( \hat{P}_{k} \),
which is illustrated in~\autoref{alg/safety/verif}.
%
Let \( p \in (0, 1) \)
denote a desired probability threshold.
Then,
for each \( i \in 1, 2, \ldots, T \),
firstly we apply \autoref{alg/safety/merging} iteratively until
no cells remain in the new graph \( \mathcal{D}' \) which can be merged
(lines
\autoref{algline/safety/verif/merge/start}-%
\autoref{algline/safety/verif/merge/end}).
Finally,
for each cell \( \mathcal{S}_j \) of the original graph \( \mathcal{D} \),
we compute the safety probability bound \( \hat{P}_{i}(\mathcal{S}_j) \) using
\eqref{eq/safety/tpn/0}
(lines
\autoref{algline/safety/verif/propagate/start}-%
\autoref{algline/safety/verif/propagate/end}).
As such,
by virtue of \autoref{stmt/safety/merging} and \autoref{stmt/safety/tpn},
the bounds \( \hat{P}_{T} \) obtained using \autoref{alg/safety/verif} are
valid upper bounds on the safety probability and are guaranteed to be
at least as tight as the ones obtained by applying
\eqref{eq/spb/naive} on \( \mathcal{D} \).
We validate this result in~\autoref{sec/sim}.

\begin{algorithm}[t]
  \caption{Verifcation Framework}%
  \label{alg/safety/verif}
  \hspace*{0.02in} {\bf Input:} \(\mathcal{D}\), $\hat{P}_{0}$, $p$  \\
  \hspace*{0.02in} {\bf Output:} $\hat{P}_{T}$
  \begin{algorithmic}[1]
    \FOR{$i \text{ in } 1, 2, \ldots, T$}
    \STATE $\mathcal{D'} \gets \mathcal{D}$ %
    \label{algline/safety/verif/merge/start}
    \REPEAT
    \STATE $\mathcal{D''} \gets \mathcal{D'}$ %
    \FOR{$j \text{ in } 1, 2, \ldots, |\mathcal{S}'|$}
    \STATE $\mathcal{D''}, \hat{P}_{i-1} \gets$ %
    MergeCells($\mathcal{D''}, \mathcal{S}'_j, p, \hat{P}_{i-1}$)
    \ENDFOR
    \STATE flag $ \gets \mathcal{D}'' = \mathcal{D}' $
    \STATE $ \mathcal{D}' \gets \mathcal{D}'' $
    \UNTIL{ \textbf{not} flag }
    \label{algline/safety/verif/merge/end}
    \FOR{$j \text{ in } 1, 2, \ldots, |\mathcal{S}|$}
    \label{algline/safety/verif/propagate/start}
    \STATE $\hat{P}_{i}(\mathcal{S}'_j) \gets $ %
    PropagateSafetyProb($\mathcal{S}_j, \mathcal{N}'_{\mathcal{S}_{j}}, \hat{P}_{i-1}$)
    \label{algline/safety/verif/propagate/end}
    \ENDFOR
    \ENDFOR
  \end{algorithmic}
\end{algorithm}



\section{SNN-SMC based Refinement}%
\label{sec/refine}

\begin{figure}[]
  \centering
  \begin{minipage}[t]{1\linewidth}
    \centering
    \includegraphics[height=4.0cm]{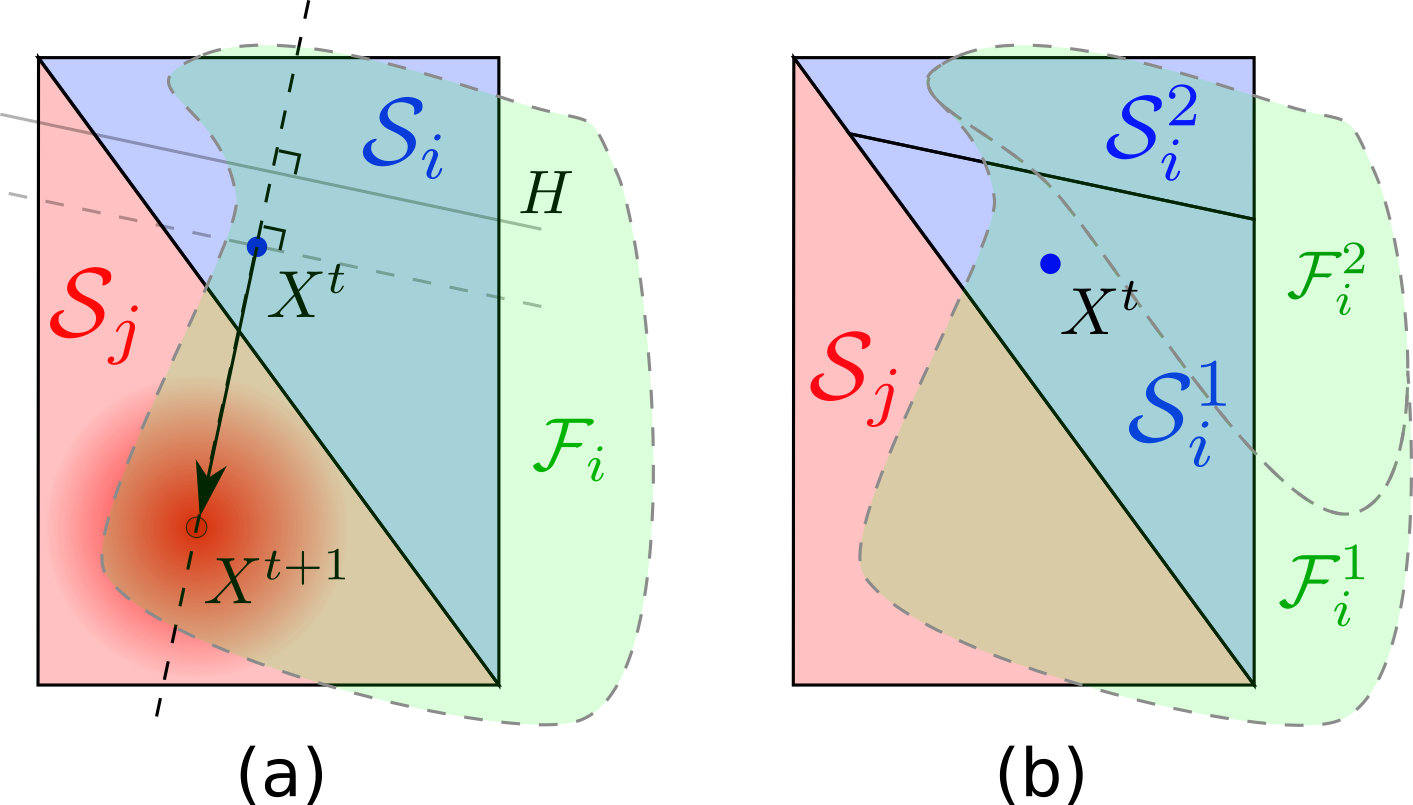}
  \end{minipage}%
  \caption{
    Subdivision of cell \( \mathcal{S}_i \) w.r.t.
    the transition probability \( P(\mathcal{S}_i, \mathcal{S}_j) \)
    into two subcells \( \mathcal{S}_i^{1} \) and \( \mathcal{S}_i^{2} \).
    The forward reachable sets
    \( \mathcal{F}_i, \mathcal{F}_{i}^{1}, \mathcal{F}_{i}^2 \)
    of cells \( \mathcal{S}_i, \mathcal{S}_i^{1}, \mathcal{S}_i^{2} \)
    under the dynamics of the closed-loop system
    \( X^{t+1} = A X^{t} + B f_{\mathrm{NN}}(d^{t}(X^{t})) \).
    are depicted in green,
    respectively.
    We notice that the selection of hyperplane \( H \) minimizes
    the probability of a state \( X^{t} \in \mathcal{S}_{i}^{2} \)
    to land in \( \mathcal{S}_j \),
    by placing a sufficiently large neighborhood of
    \( X^{t} \in \mathcal{X}^{\star}_{ij} \) strictly inside
    \( X^{t} \in \mathcal{S}_{i}^{1} \).
  }
  \label{fig/refinement}
\end{figure}

In this section,
we present a method to refine the cells of
a given partition $ \mathcal{S} $
in order to obtain tighter upper bounds on
the safety probabilities \( \hat{P}_k \)
compared to those obtained for a coarser initial partition.
%
We begin by presenting the following proposition which
provides bounds on the transition probabilities of
the cells \( \mathcal{S}_i^{1}, \mathcal{S}_i^{2} \)
obtained by cutting a cell \( \mathcal{S}_i \in \mathcal{S} \)
into two disjoint cells separated by a hyperplane \( H \).
%
%
%
%
\begin{proposition}
  Given
  the transition graph \(\mathcal{D}\)
  and
  safety probabilities \( \hat{P}_k \),
  let \( H \) be a hyperplane splitting \(\mathcal{S}_i \in \mathcal{S} \)
  into \(\mathcal{S}^1_i\) and \(\mathcal{S}^2_i\).
  For all \( \mathcal{S}_j \in \mathcal{N}_{\mathcal{S}_{i}} \),
  if
  \(
  \hat{P}(\mathcal{S}^1_i, \mathcal{S}_j)
  =
  \hat{P}(\mathcal{S}_i, \mathcal{S}_j)
  \),
  then
  \(
  \hat{P}(\mathcal{S}^2_i, \mathcal{S}_j)
  \leq
  \hat{P}(\mathcal{S}_i, \mathcal{S}_j)
  \).
  Also,
  the transition graph
  \(
  \mathcal{D}'
  =
  (\mathcal{S}', \mathcal{E}', \hat{P}')
  \)
  is valid,
  where
  \( \mathcal{S}' \) and \( \mathcal{E}' \)
  are obtained by replacing the cell \( \mathcal{S}_i \)
  with \( \mathcal{S}^1_i \) and \( \mathcal{S}^2_i \).
  Moreover,
  \(
  \hat{P}_{k+1}(\mathcal{S}^{l}_i)\leq \hat{P}_{k+1}(\mathcal{S}_i)
  \)
  for all \( l \in \{1, 2\} \).
\end{proposition}
In words,
given \( \mathcal{S}_j \in \mathcal{N}_{\mathcal{S}_i} \),
the transition probability bound
\( \hat{P}(\mathcal{S}_i^{\ell}, \mathcal{S}_j) \)
computed using \autoref{alg/prob/transition} is the same as
\( \hat{P}(\mathcal{S}_i, \mathcal{S}_j) \)
for at least one \( \ell \in \{1, 2\} \).
The reason is that satisfiability of the SNN-SMC problem
in \autoref{alg/prob/transition}
implies that there exists at least one
\(
X^t \in \mathcal{S}_i = \mathcal{S}_i^{1} \cup \mathcal{S}_i^{2}
\)
that marginally satisfies
the inequality
\(
P(x^{t+1} \in \mathcal{S}_j | X^t)
\geq
\hat{P}(\mathcal{S}_i, \mathcal{S}_j) - dq
\).
Let \( \mathcal{X}^{\star}_{ij} \) denote the set of all these \( X^{t} \).
Assuming that \( \mathcal{X}^{\star}_{ij} \) lies strictly
in the interior of \( \mathcal{S}_i \),
it can be placed in one of the subcells (e.g., \( \mathcal{S}_i^1 \))
through a proper selection of the hyperplane \( H \)
so that
the transition probability of the other subcell
(e.g.,
\(
\hat{P}(\mathcal{S}_i^{2}, \mathcal{S}_j)
\))
becomes strictly less than
\( \hat{P}(\mathcal{S}_i, \mathcal{S}_j) \)
(see \autoref{fig/refinement}).
%
%
To find this hyperplane,
we consider a state \( X^{t} \in \mathcal{X}^{\star}_{ij} \)
that marginally satisfies the SNN-SMC problem for
\( \mathcal{S}_i, \mathcal{S}_j \)
and
define the hyperplane \( H \)
that is perpendicular to the line connecting
\( X^{t}, X^{t+1} \)
and
contains \( X^{t} \),
where
\( X^{t+1} = A X^{t} + B f_{\mathrm{NN}}(d(X^t)) \).
%
Assuming that all \( X^{t} \in \mathcal{X}^{\star}_{ij} \) lie on
the same half-space defined by \( H \),
translating this hyperplane away from \( \mathcal{X}^{\star}_{ij} \)
will decrease the transition probability bound
from one of the new cells \( \mathcal{S}_i^1 \) or \( \mathcal{S}_j^2 \)
to \( \mathcal{S}_j \).
Thus,
given a pair of cells
\( \mathcal{S}_i \in \mathcal{S} \)
and
\( \mathcal{S}_j \in \mathcal{N}_{\mathcal{S}_i} \),
we can refine the partition \(S\) by translating the hyperplane \(H\) so that
either
\( \hat{P}(\mathcal{S}_i^1, \mathcal{S}_j) \)
or
\( \hat{P}(\mathcal{S}_i^2, \mathcal{S}_j) \)
is minimized.

\section{Numerical Experiments}%
\label{sec/sim}

In this section,
we present simulation results
to validate the proposed bounds on the probability that
a point-sized robot collides with
the boundary of the non-convex planar workspace \( \mathcal{W} \).
%
%
Particularly,
we consider a scenario similar to the one in~\cite{b1}
and
assume that the robot's dynamics can be modeled as a single-integrator, i.e.:
\begin{equation}
  x_{t+1} = x_{t} + f_{NN}(d(x_t)) + w^{t},
\end{equation}
where
$x^t\in\mathbb{R}^2$ denotes the robot's position at time step \(t\)
and
$w^t \sim \mathcal{N}(0, 3)$.
%
%
\begin{figure}[]
  \centering
  \begin{minipage}[t]{1\linewidth}
    \centering
    \includegraphics[width=8cm,height=6.0cm]{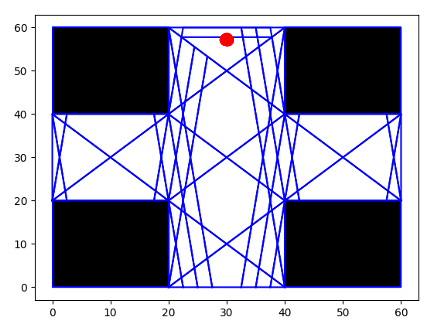}
  \end{minipage}%
  \caption{
    Compact workspace \(\mathcal{W}\) and partition \( \mathcal{S} \).
    Areas shaded in black correspond to obstacles whereas blue lines indicate
    the boundaries of the convex cells.
    The red dot marks the goal position of the robot.
  }
  \label{fig/sim/ws}
\end{figure}
Additionally,
we assume that the robot is equipped with a LiDAR scanner that
emits a set of \(q\) lasers evenly distributed in a \(2 \pi ~ rad\) fan,
i.e.,
\(
d(x^t) = [(d_0(x^t))^T, (d_1(x^t))^T, ..., (d_q(x^t))^T]^T,
\)
where
\(
d_i(x^t) = [r_i(x^t)\cos{\theta_i}, r_i(x^t)\sin{\theta_i}]^T
\)
and
\( r_i(x^t) \)
denotes the distance measured between
the robot and
the closest obstacle in the direction \( [\cos(\theta_i), \sin(\theta_i)] \),
for all \(i \in \{1, 2, \ldots q\}\).
To drive the robotic system
to a predetermined goal position using
only the feedback \( d(x^t) \),
we employed a ReLU neural network controller \(f_{\mathrm{NN}}\)
consisting of three hidden layers and a total of 32 neurons.
Lastly,
we used the partitioning method proposed in~\cite{b1}
to partition the domain \( \mathcal{X} \),
which in this scenario coincides with the workspace \( \mathcal{W} \).
This partition is depicted in~\autoref{fig/sim/ws}.


\begin{figure*}
  \centering
  \subfigure[]{
    \begin{minipage}[h]{0.3\linewidth}
      \centering
      \includegraphics[width = 5cm]{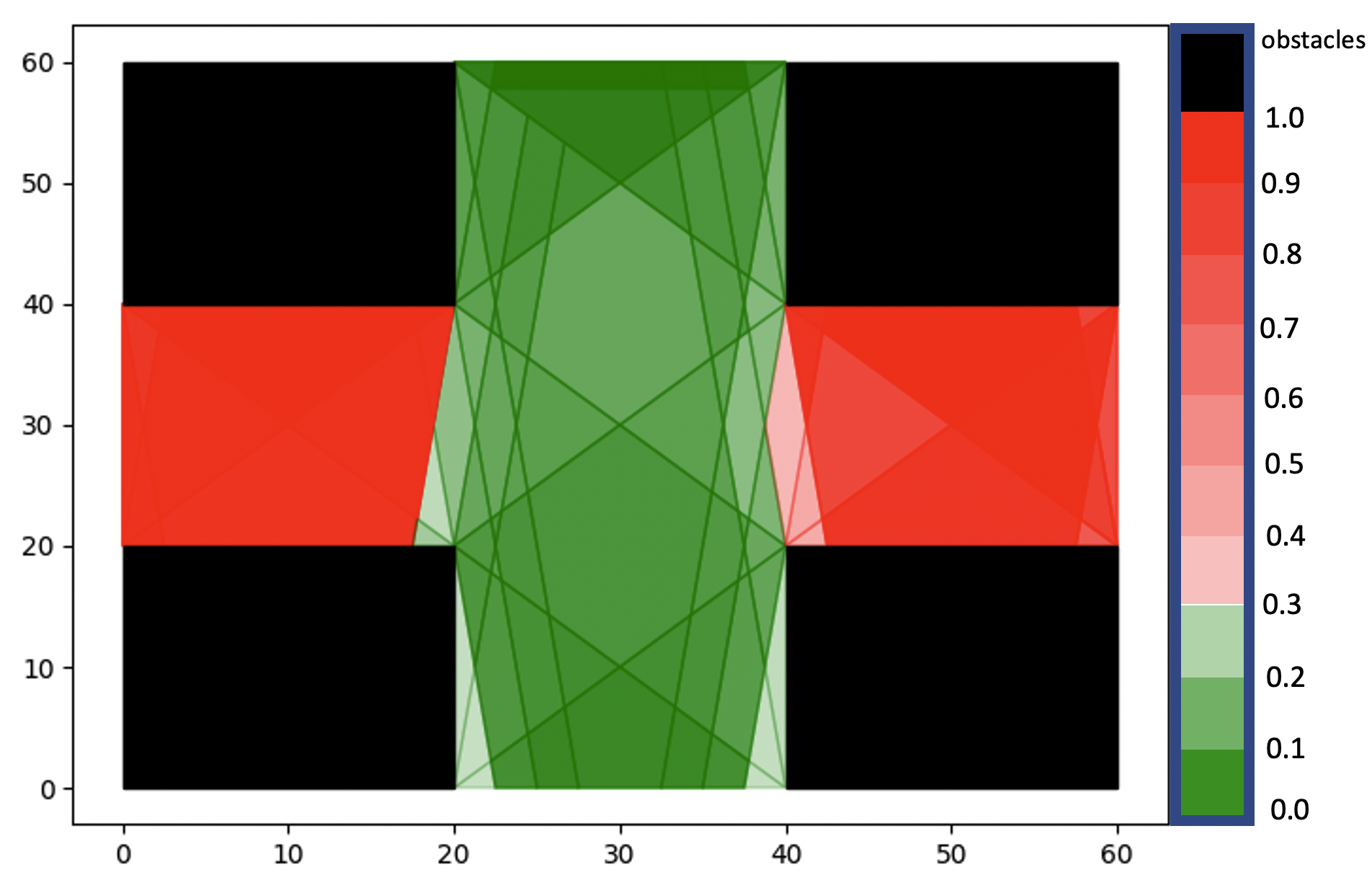}
    \end{minipage}
  }
  \subfigure[]{
    \begin{minipage}[h]{0.3\linewidth}
      \centering
      \includegraphics[width = 5cm]{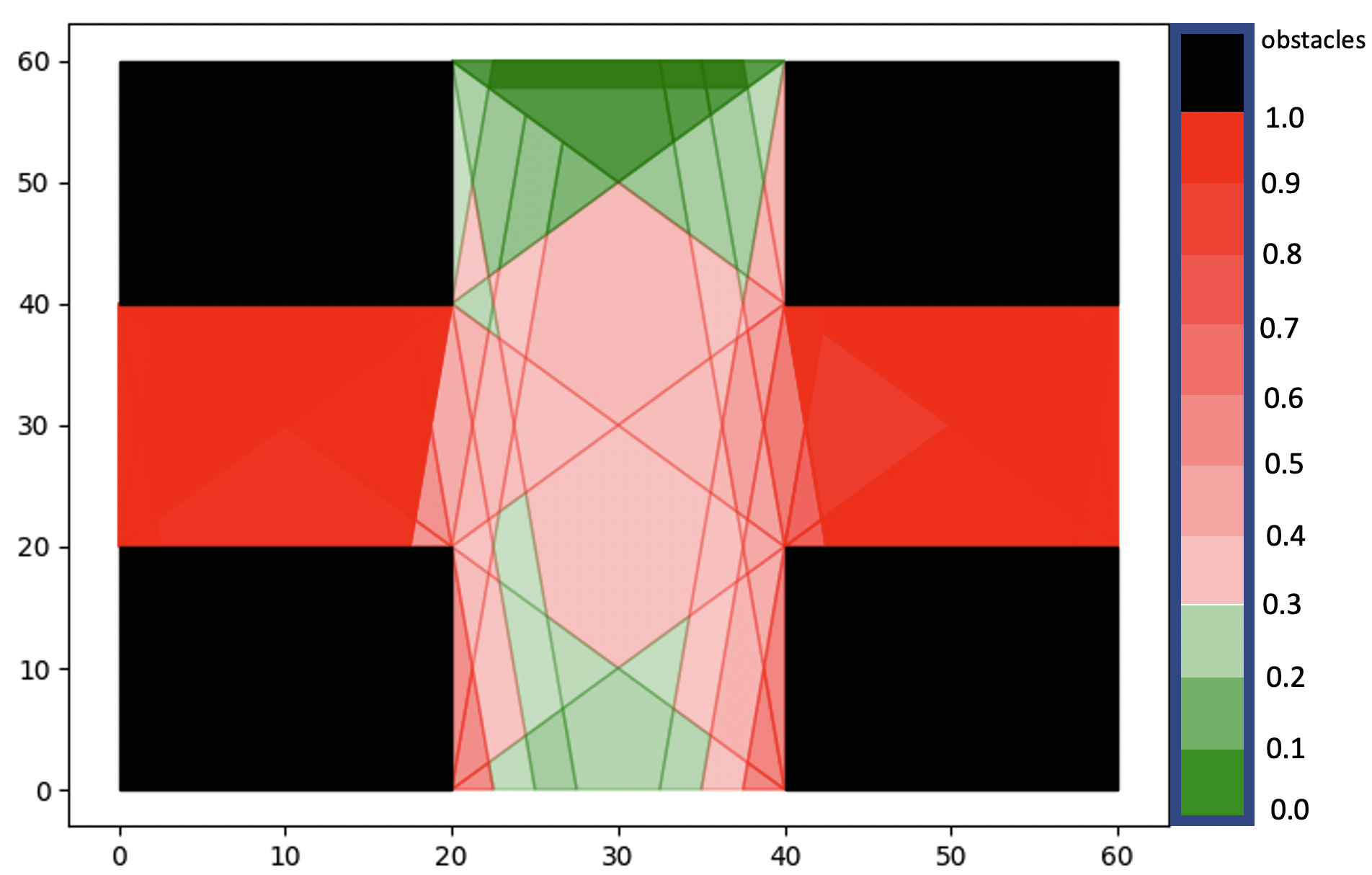}
    \end{minipage}
  }
  \subfigure[]{
    \begin{minipage}[h]{0.3\linewidth}
      \centering
      \includegraphics[width = 5cm]{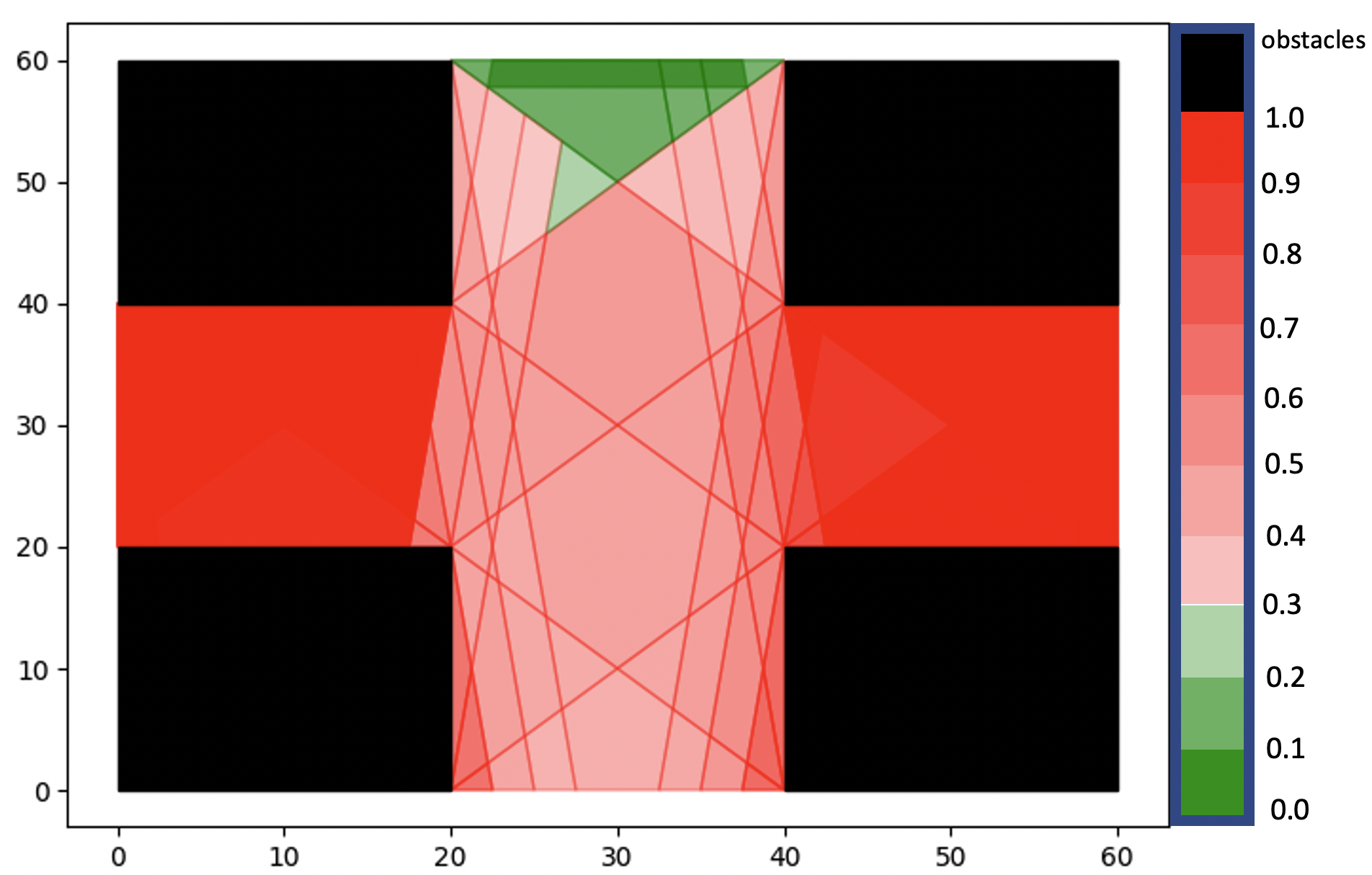}
    \end{minipage}
  }

  \centering
  \subfigure[]{
    \begin{minipage}[h]{0.3\linewidth}
      \centering
      \includegraphics[width = 5cm]{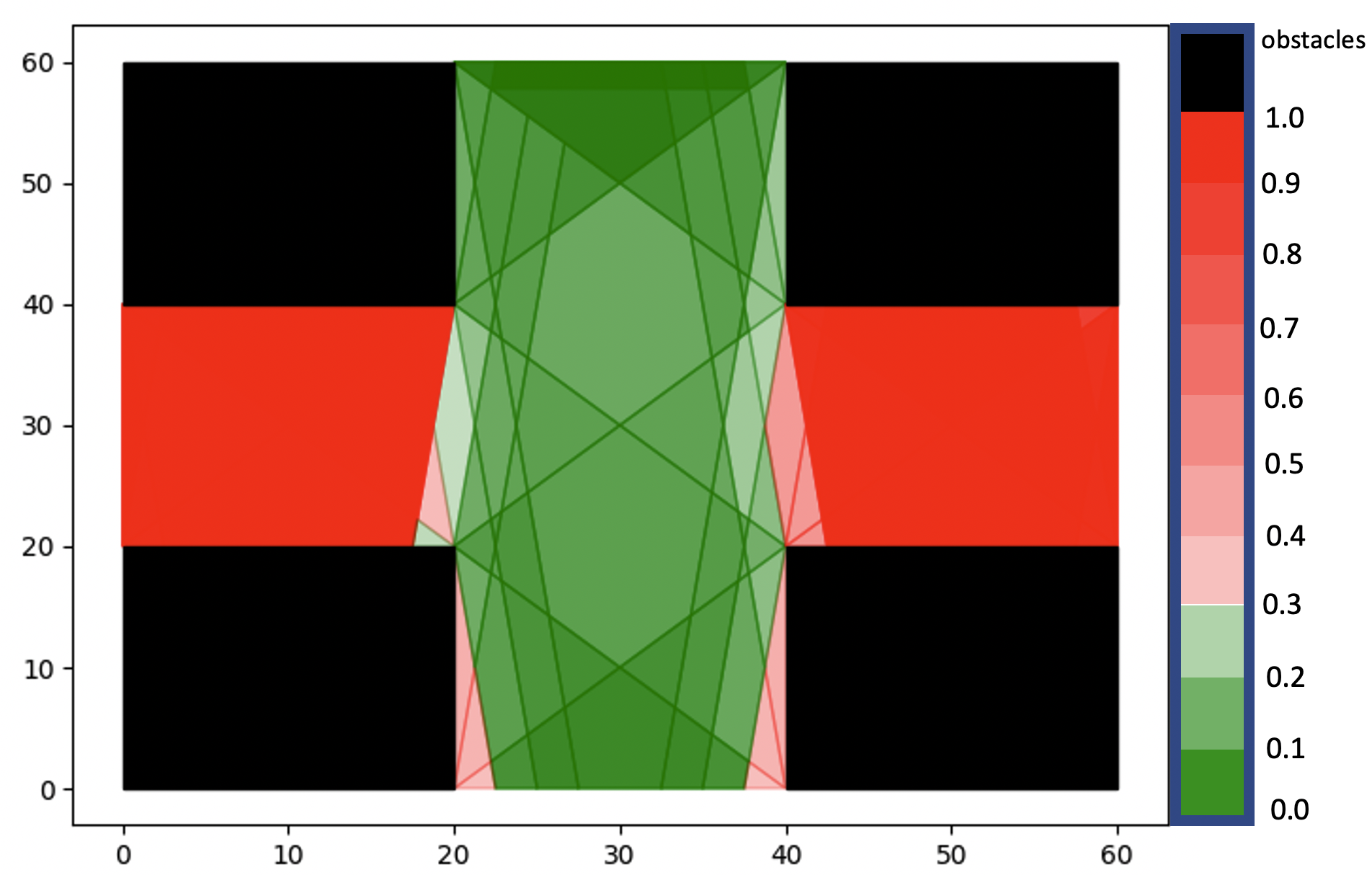}
    \end{minipage}
  }
  \subfigure[]{
    \begin{minipage}[h]{0.3\linewidth}
      \centering
      \includegraphics[width = 5cm]{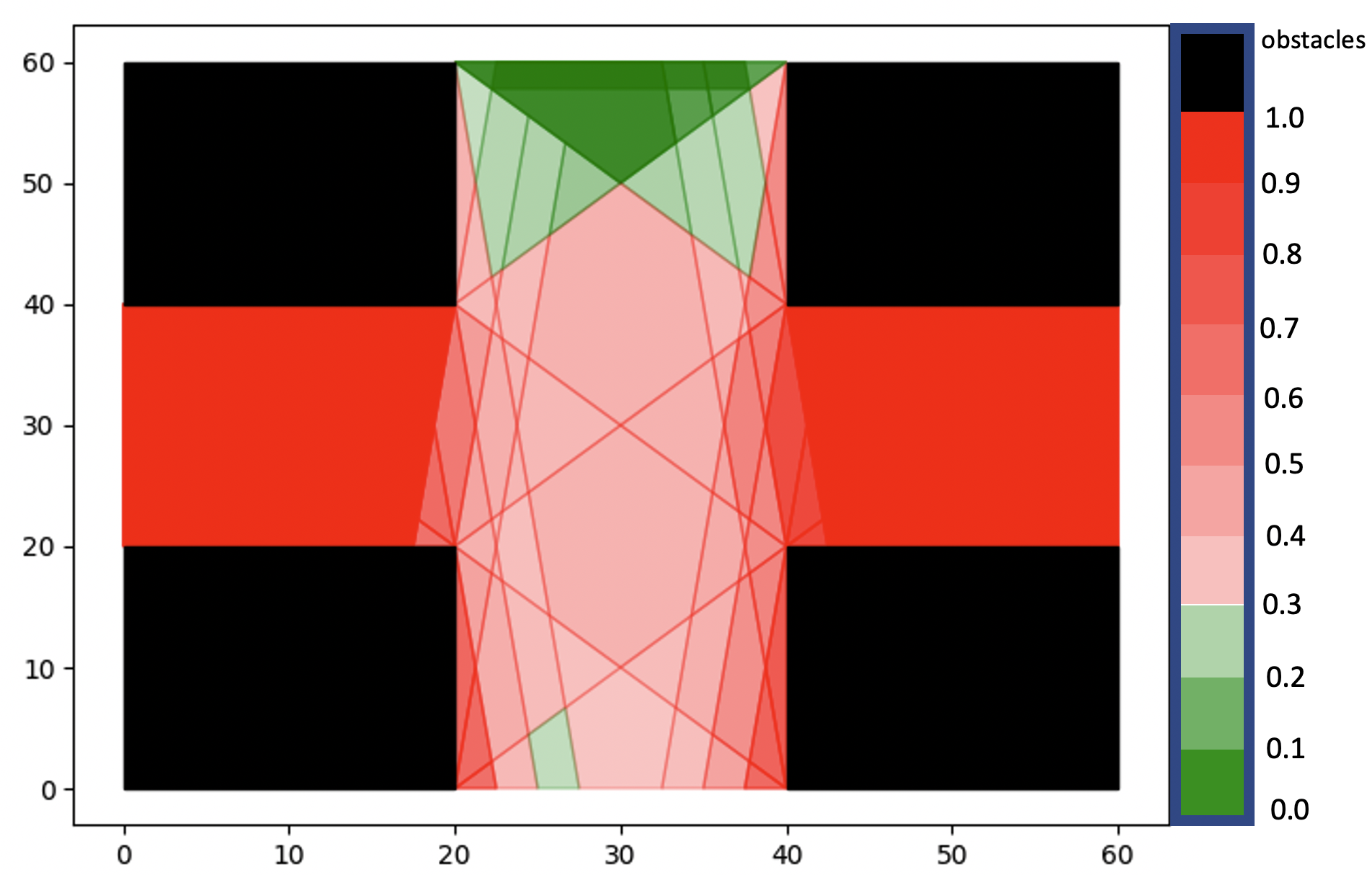}
    \end{minipage}
  }
  \subfigure[]{
    \begin{minipage}[h]{0.3\linewidth}
      \centering
      \includegraphics[width = 5cm]{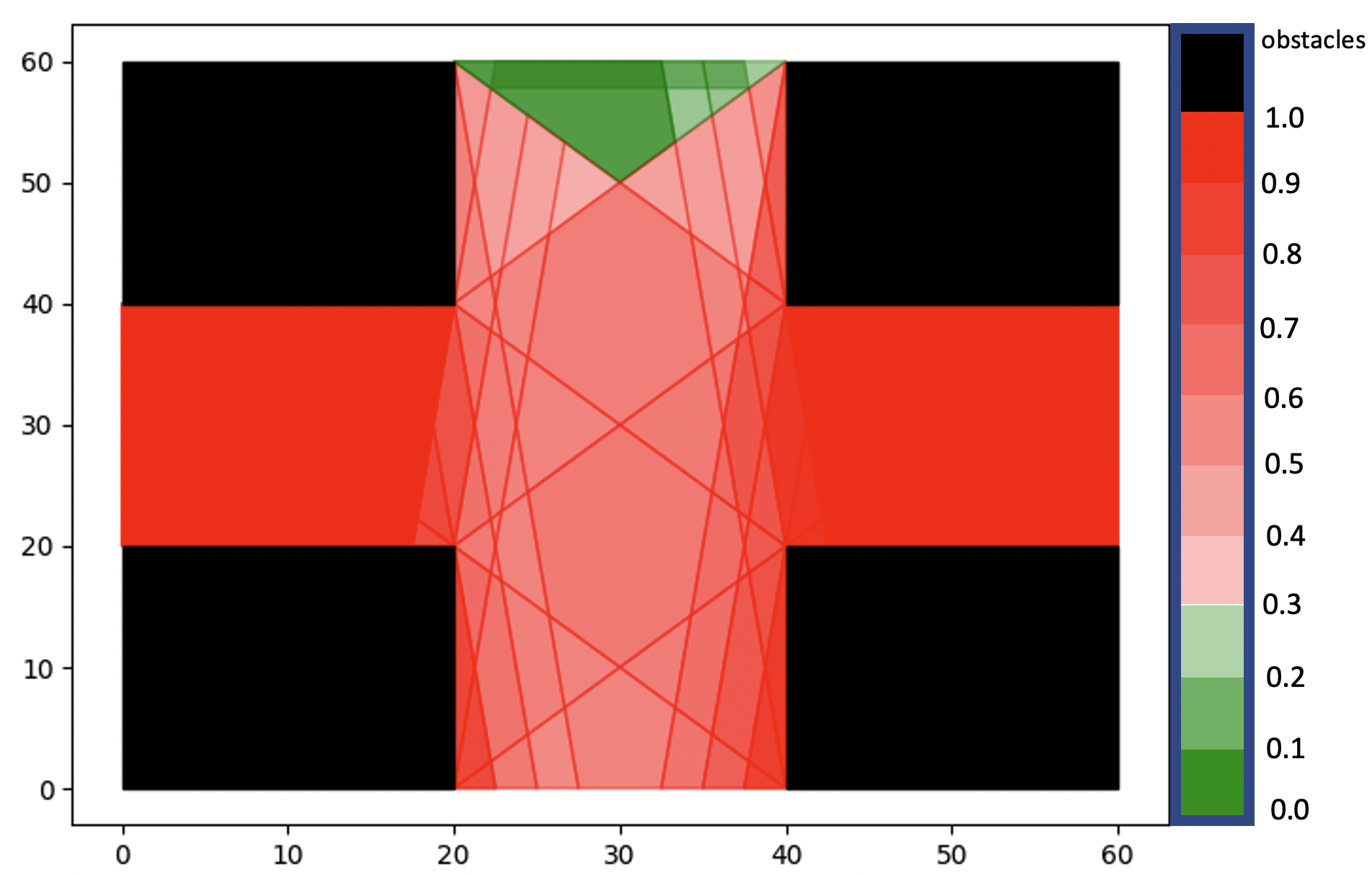}
    \end{minipage}
  }
  \caption{
    Safety probability bound \( \hat{P}_{T} \) for various horizons
    computed using
    \autoref{alg/safety/verif} ((a), (b) and (c))
    and
    the algorithm proposed in \cite{b20} ((d), (e), (f)).
    Left column corresponds to \( T = 3\),
    middle column to \( T = 6 \) and
    right column to \( T = 9 \).
  }
  \label{fig/cmpr}
\end{figure*}

\begin{figure}[t]
  \centering
  \subfigure[Without Merging and Refinement]{
    \begin{minipage}[h]{0.45\linewidth}
      \centering
      \includegraphics[width = 4cm]{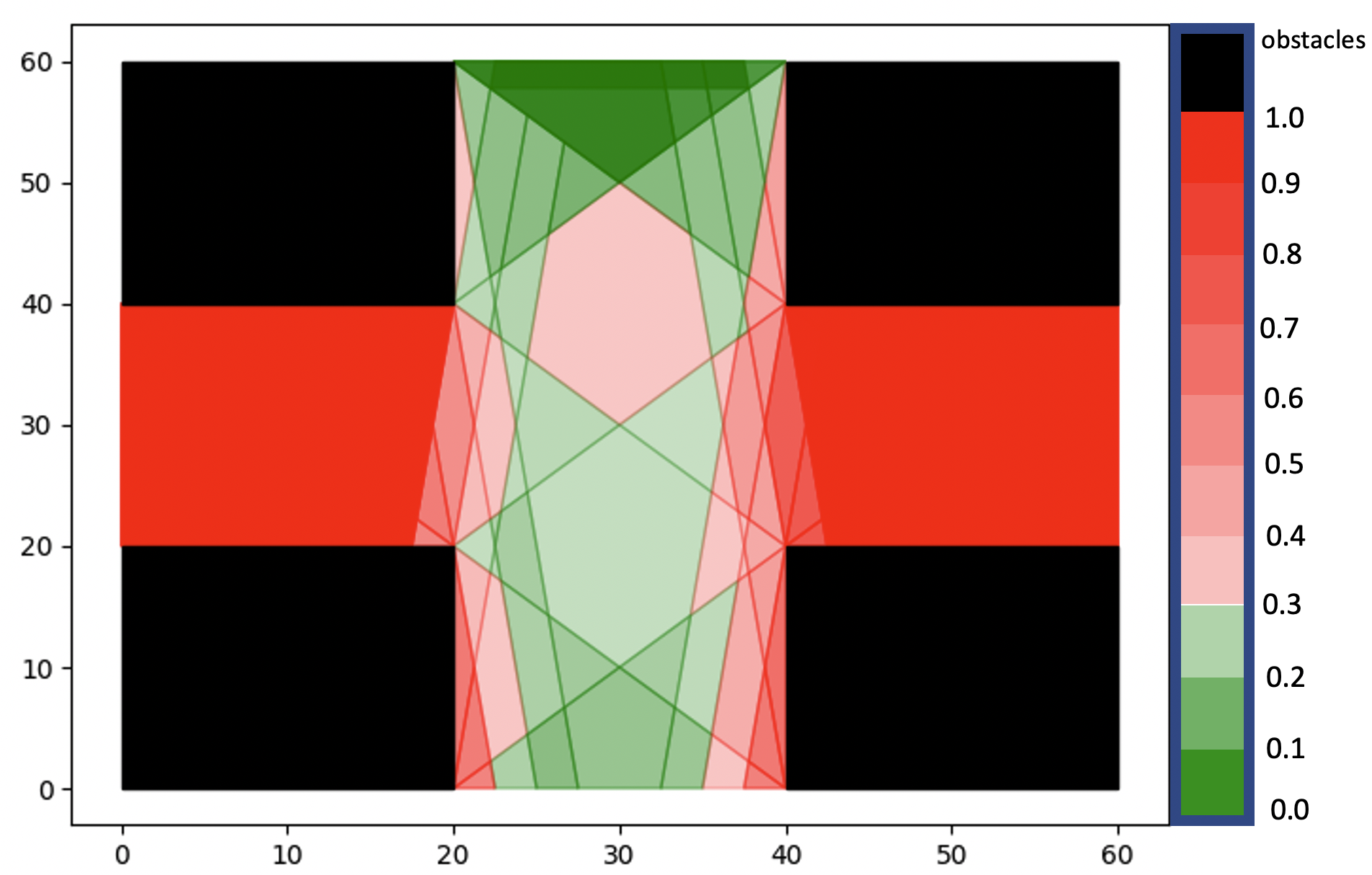}
    \end{minipage}
  }
  \subfigure[With Merging]{
    \begin{minipage}[h]{0.45\linewidth}
      \centering
      \includegraphics[width = 4cm]{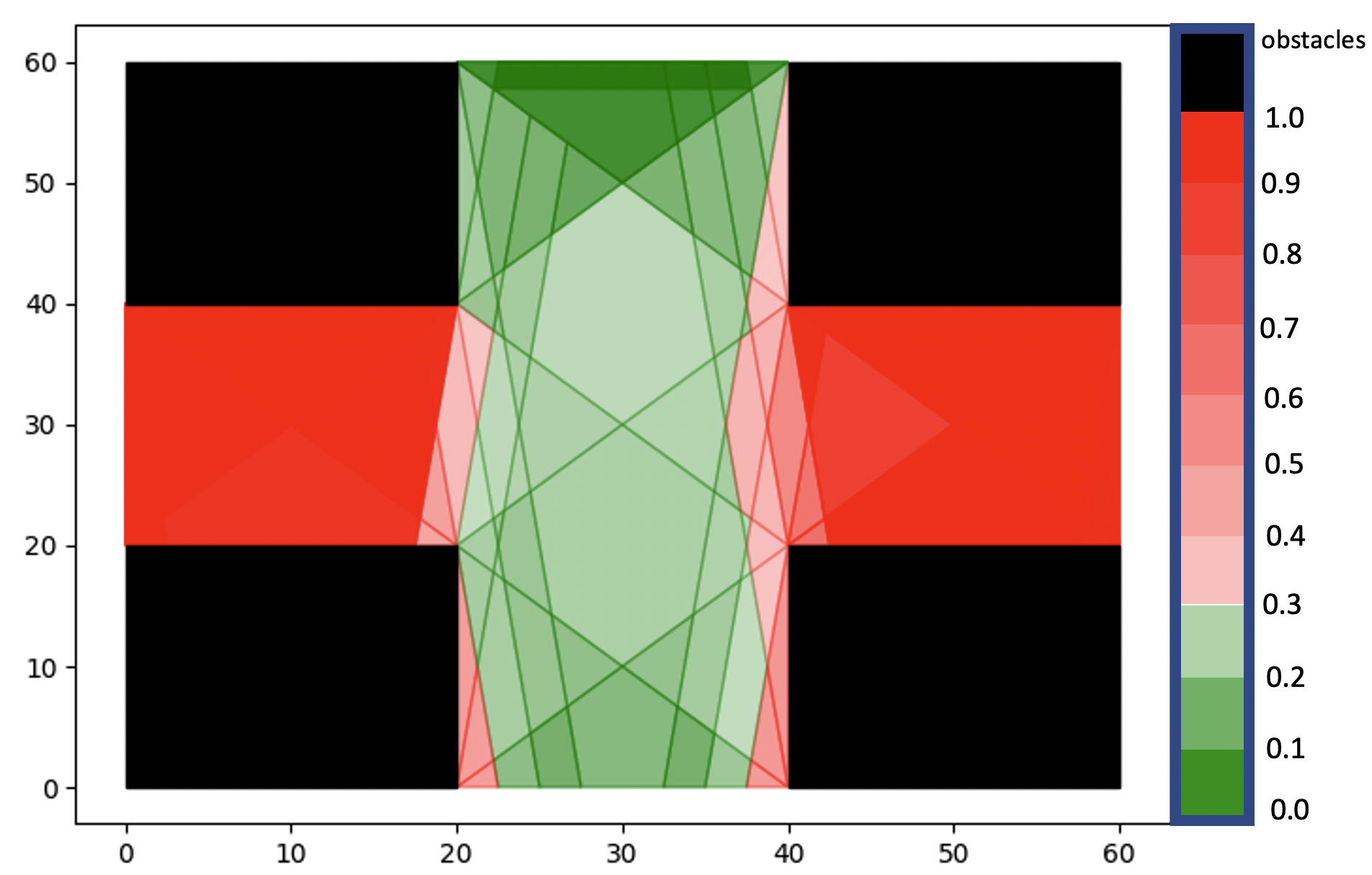}
    \end{minipage}
  }
  \centering
  \subfigure[With Refinement]{
    \begin{minipage}[h]{0.45\linewidth}
      \centering
      \includegraphics[width = 4cm]{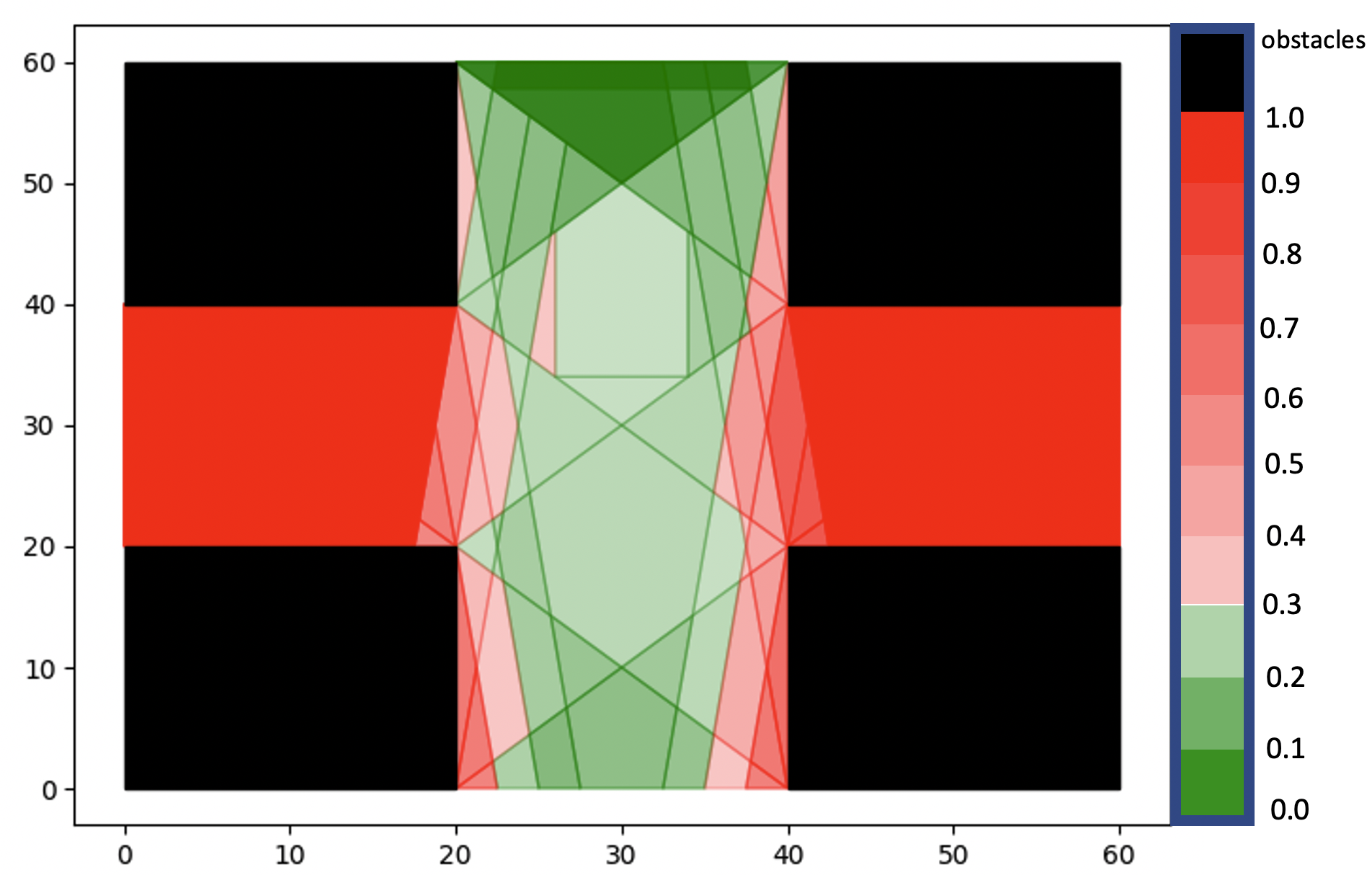}
    \end{minipage}
  }
  \subfigure[With Merging and Refinement]{
    \begin{minipage}[h]{0.45\linewidth}
      \centering
      \includegraphics[width = 4cm]{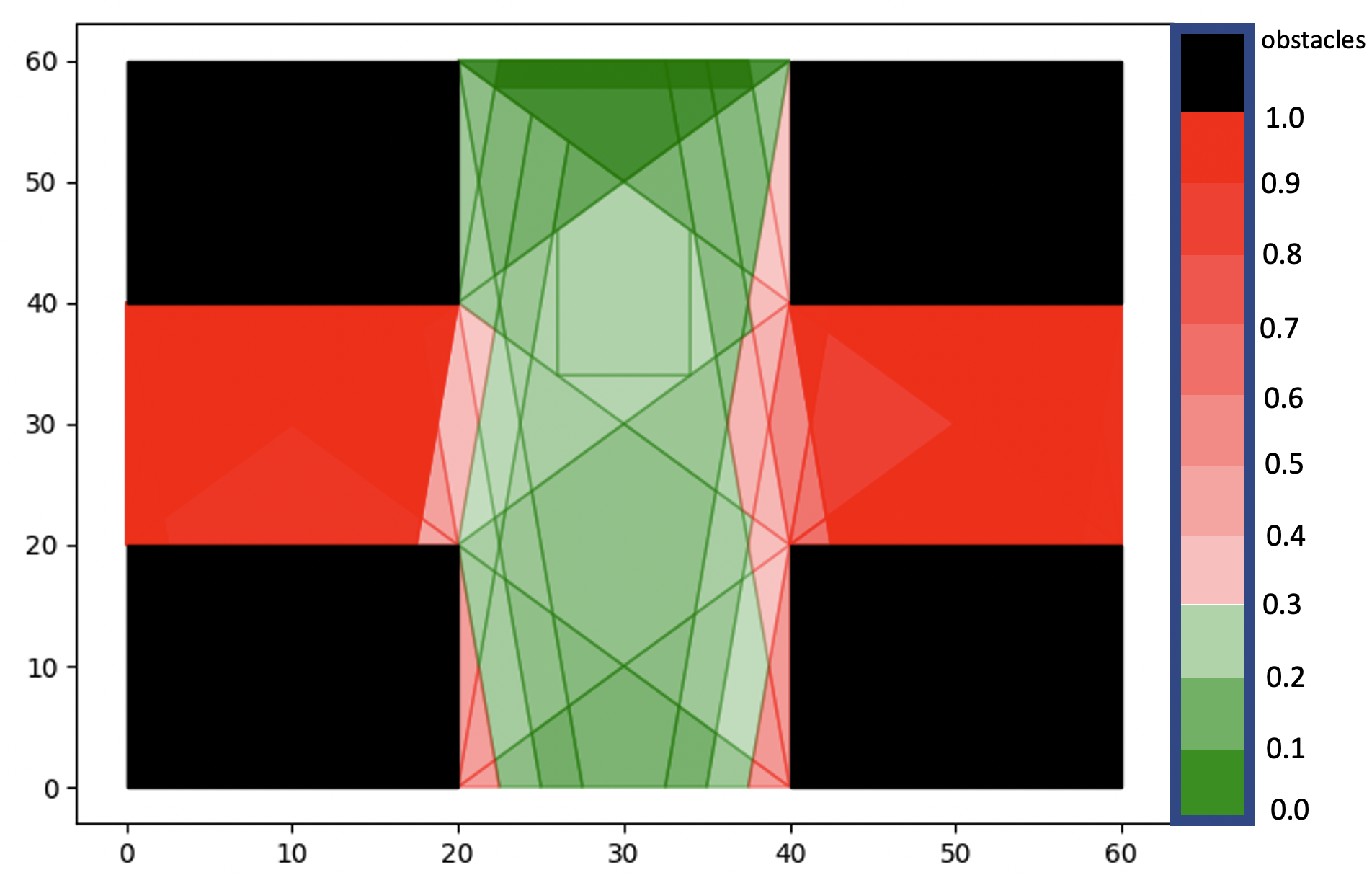}
    \end{minipage}
  }
  \caption{
    Safety probability bound \( \hat{P}_{6} \)
    computed using
    \eqref{eq/spb/naive} (left column)
    and
    \autoref{alg/safety/verif}
    for horizon $T=7$ on the original and refined graphs, respectively.
  }
  \label{fig/var/part}
  \vspace{-1cm}
\end{figure}
\begin{figure}[h]
  \centering
  \begin{minipage}[t]{1\linewidth}
    \centering
    \includegraphics[width=8cm]{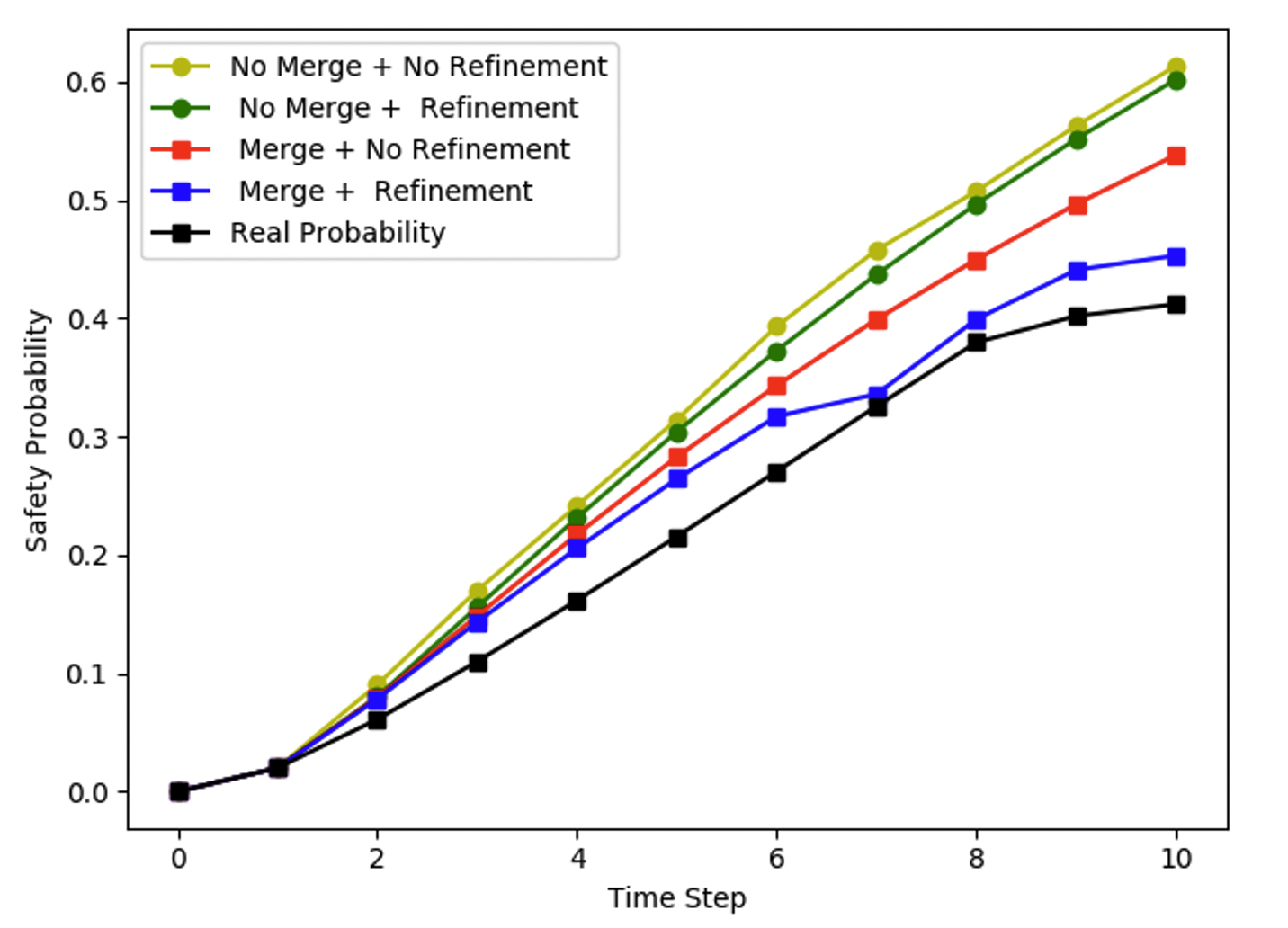}
  \end{minipage}%
  \caption{
    Bound on and actual safety probability \( P_{k} \)
    of a given cell \( \mathcal{S}_i \)
    for different horizons.
  }
  \label{fig/var/cell}
\end{figure}
In~\autoref{fig/cmpr}
we provide a comparison between
the bound \( \hat{P}_{T} \) computed over the given domain
using~\autoref{alg/safety/verif}
and
the method proposed in~\cite{b20},
for various horizons \( T \).
We observe that
our method provides tighter bounds compared to~\cite{b20} for every horizon,
even for the coarse partition considered in this scenario.
Additionally,
in~\autoref{fig/var/part},
we show the bounds on the safety probability \( P_{6} \)
computed over this domain using
\eqref{eq/spb/naive} and \autoref{alg/safety/verif} on
the original and refined transition graphs,
respectively.
%
The refined graph was obtained using
the heuristic subdivision scheme presented in~\autoref{sec/refine} to
the largest cell located near the center of the workspace.
It can be seen in \autoref{fig/var/part} that
the safety probability bounds estimated using
merging and/or refinement are noticeably tighter than the ones without.
Finally,
we certify the correctness of the proposed safety probability bounds by
comparing them to the true safety probability of
for the cell that is adjacent to the one subject to the refinement.
To estimate the true safety probability,
we simulate a sufficiently large number (\(\approx 10000\) of
robot trajectories starting from states within that cell
and
compute the percentage of those that end up violating the safety requirements
as a result of the applied disturbances.
In~\autoref{fig/var/cell}
we present the safety probability bounds returned by our method for
that given cell for different horizons and
compare these bounds to the estimated true safety probability.
We observe that all bounds returned by our method correctly upper bound
the true probability,
while the bounds obtained by using both merging and refinement
are the tightest.
We also remark that
the gap between the estimated bounds and the true safety probability
becomes larger as the horizon increases.



\section{Conclusions}%
\label{sec/conclusions}

In this work,
we addressed the problem of formal safety verification of
stochastic cyber-physical systems (CPS) equipped with
a ReLU neural network (NN) controllers.
Particularly,
we presented a method to compute sets of initial states which
that ensure that the system trajectories are safe 
within a specified horizon.
To do this,
we designed a suitable discrete abstraction of the system
and
formulated an SMC problem to estimate upper bounds on
the transition probabilities between cells in this discrete abstraction.
Additionally,
we proposed a method to obtain tighter bounds on
the corresponding safety probability
as well as a heuristic for refining the abstraction in a way that
may further improve the results.
Finally,
we presented simulation results verifying the efficacy of our method
compared to existing methodologies proposed in the literature.



\bibliographystyle{IEEEtran}
\bibliography{references.bib}

\end{document}